\documentclass[lettersize,journal]{IEEEtran}
\IEEEoverridecommandlockouts
\usepackage{amsmath,amssymb,amsfonts}
\usepackage{cite}
\usepackage{graphicx}
\usepackage{epstopdf}
\usepackage{amssymb}
\usepackage{textcomp}
\usepackage{times}
\usepackage{subfigure}
\usepackage{amssymb,amsmath}
\usepackage{acronym}  
\usepackage{balance}
\usepackage{lettrine}
\usepackage{bm}
\usepackage{algorithm}
\usepackage{algorithmic}
\usepackage{mathrsfs}
\usepackage{lipsum}
\usepackage{stfloats}
\usepackage{placeins}
\usepackage{tabularx}
\usepackage{booktabs}
\usepackage{multirow}
\usepackage{url}
\usepackage{graphicx}

\allowdisplaybreaks[4]
\renewcommand{\IEEEQED}{\IEEEQEDopen} 

\usepackage{bm}

\DeclareMathAlphabet{\mathsfbr}{OT1}{cmss}{m}{n}
\SetMathAlphabet{\mathsfbr}{bold}{OT1}{cmss}{bx}{n}
\DeclareRobustCommand{\msf}[1]{%
  \ifcat\noexpand#1\relax\msfgreek{#1}\else\mathsfbr{#1}\fi
}

\makeatletter
\newcommand{\msfgreek}[1]{\csname s\expandafter\@gobble\string#1\endcsname}
\makeatother

\DeclareFontEncoding{LGR}{}{} 
\DeclareSymbolFont{sfgreek}{LGR}{cmss}{m}{n}
\SetSymbolFont{sfgreek}{bold}{LGR}{cmss}{bx}{n}
\DeclareMathSymbol{\salpha}{\mathord}{sfgreek}{`a}
\DeclareMathSymbol{\sbeta}{\mathord}{sfgreek}{`b}
\DeclareMathSymbol{\sgamma}{\mathord}{sfgreek}{`g}
\DeclareMathSymbol{\sdelta}{\mathord}{sfgreek}{`d}
\DeclareMathSymbol{\sepsilon}{\mathord}{sfgreek}{`e}
\DeclareMathSymbol{\szeta}{\mathord}{sfgreek}{`z}
\DeclareMathSymbol{\seta}{\mathord}{sfgreek}{`h}
\DeclareMathSymbol{\stheta}{\mathord}{sfgreek}{`j}
\DeclareMathSymbol{\siota}{\mathord}{sfgreek}{`i}
\DeclareMathSymbol{\skappa}{\mathord}{sfgreek}{`k}
\DeclareMathSymbol{\slambda}{\mathord}{sfgreek}{`l}
\DeclareMathSymbol{\smu}{\mathord}{sfgreek}{`m}
\DeclareMathSymbol{\snu}{\mathord}{sfgreek}{`n}
\DeclareMathSymbol{\sxi}{\mathord}{sfgreek}{`x}
\DeclareMathSymbol{\somicron}{\mathord}{sfgreek}{`o}
\DeclareMathSymbol{\spi}{\mathord}{sfgreek}{`p}
\DeclareMathSymbol{\srho}{\mathord}{sfgreek}{`r}
\DeclareMathSymbol{\ssigma}{\mathord}{sfgreek}{`s}
\DeclareMathSymbol{\stau}{\mathord}{sfgreek}{`t}
\DeclareMathSymbol{\supsilon}{\mathord}{sfgreek}{`u}
\DeclareMathSymbol{\sphi}{\mathord}{sfgreek}{`f}
\DeclareMathSymbol{\schi}{\mathord}{sfgreek}{`q}
\DeclareMathSymbol{\spsi}{\mathord}{sfgreek}{`y}
\DeclareMathSymbol{\somega}{\mathord}{sfgreek}{`w}

\DeclareMathSymbol{\svarsigma}{\mathord}{sfgreek}{`c}

\DeclareMathSymbol{\sGamma}{\mathalpha}{sfgreek}{`G}
\DeclareMathSymbol{\sDelta}{\mathalpha}{sfgreek}{`D}
\DeclareMathSymbol{\sTheta}{\mathalpha}{sfgreek}{`J}
\DeclareMathSymbol{\sLambda}{\mathalpha}{sfgreek}{`L}
\DeclareMathSymbol{\sXi}{\mathalpha}{sfgreek}{`X}
\DeclareMathSymbol{\sPi}{\mathalpha}{sfgreek}{`P}
\DeclareMathSymbol{\sSigma}{\mathalpha}{sfgreek}{`S}
\DeclareMathSymbol{\sUpsilon}{\mathalpha}{sfgreek}{`U}
\DeclareMathSymbol{\sPhi}{\mathalpha}{sfgreek}{`F}
\DeclareMathSymbol{\sPsi}{\mathalpha}{sfgreek}{`Y}
\DeclareMathSymbol{\sOmega}{\mathalpha}{sfgreek}{`W}

\DeclareRobustCommand{\mcal}[1]{%
  \ifcat\noexpand#1\relax\mathnormal{#1}\else\cal{#1}\fi
}
\DeclareRobustCommand{\BM}[1]{%
  \ifcat\noexpand#1\relax\bm{\boldUppercaseItalicGreek{#1}}\else\bm{#1}\fi
}
\makeatletter
\newcommand{\boldUppercaseItalicGreek}[1]{\csname var\expandafter\@gobble\string#1\endcsname}
\makeatother
\newcommand{\rv}[1]{\MakeLowercase{\msf{#1}}}  

\newcommand{\RM}[1]{\bm{\MakeUppercase{\msf{#1}}}}

\DeclareRobustCommand{\mcal}[1]{%
	\ifcat\noexpand#1\relax\mathnormal{#1}\else\cal{#1}\fi
}
\DeclareRobustCommand{\BM}[1]{%
	\ifcat\noexpand#1\relax\bm{\boldUppercaseItalicGreek{#1}}\else\bm{#1}\fi
}

\newtheorem{theorem}{\bf Theorem}

\newtheorem{definition}{\bf Definition}
\newtheorem{corollary}{\bf Corollary}
\newtheorem{remark}{\bf Remark}

\newtheorem{lemma}{\bf Lemma}

\acrodef{twr}[TWR]{two-way ranging}%
\acrodef{tof}[ToF]{time-of-flight}%
\acrodef{tdof}[TDoF]{time-difference-of-flight}%
\acrodef{gps}[GPS]{global positioning system}
\acrodef{stdoa}[S-TDOA]{sequential time-difference-of-arrival}%
\acrodef{tdoa}[TDOA]{time-difference-of-arrival}
\acrodef{rtt}[RTT]{round-trip time}%
\acrodef{edtwr}[NB-TWR]{network-based TWR}%
\acrodef{rtls}[RTLS]{real time location system}%
\acrodef{pr}[PR]{passive ranging}%
\acrodef{per}[PER]{passive extended ranging}%
\acrodef{nbper}[NB-PR]{network-based passive ranging}%
\acrodef{ml}[ML]{maximum likelihood}%
\acrodef{mle}[MLE]{maximum likelihood estimation}%
\acrodef{rmse}[RMSE]{root mean square error}%
\acrodef{smwr}[SMWR]{symmetric multi-way ranging}
\acrodef{smnr}[SM-NR]{signal-multiplexing network ranging}%
\acrodef{ptd}[PTD]{propagation time delay}%
\acrodef{rtls}[RTLS]{real-time location system}%
\acrodef{snr}[SNR]{signal-to-noise ratio}%
\acrodef{rtls}[RTLS]{real-time location system}%
\acrodef{wsn}[WSN]{wireless sensor network}%
\acrodef{iot}[IoT]{Internet-of-Things}%
\acrodef{mds}[MDS]{multidimensional scaling}%
\acrodef{pdf}[PDF]{probability density function}%
\acrodef{cdf}[CDF]{cumulative distribution function}%
\acrodef{los}[LOS]{line-of-sight}%
\acrodef{nlos}[NLOS]{non-line-of-sight}%

\acrodef{tx}[TX]{transmit}
\acrodef{rx}[RX]{receive}

\newcommand{\ToF}{\operatorname{ToF}}

\newcommand{\supp}{\operatorname{supp}}

\newcommand{\argmin}{\operatorname*{arg\,min}}
\newcommand{\argmax}{\operatorname*{arg\,max}}

\def\BibTeX{{\rm B\kern-.05em{\sc i\kern-.025em b}\kern-.08em
		T\kern-.1667em\lower.7ex\hbox{E}\kern-.125emX}}

\begin{document}
\title{Concurrent Coded Signal-Multiplexing Ranging \\
	for Half-Duplex Asynchronous Networks}
\author{
Zijian~Zhang,~\IEEEmembership{Graduate Student Member,~IEEE},
and Yuan~Shen,~\IEEEmembership{Senior Member,~IEEE}
\vspace*{-1em}
\\
    \thanks{
         Z. Zhang and Y. Shen are with the Department of Electronic Engineering, and Beijing National Research Center for Information Science and Technology, Tsinghua University, Beijing 100084, China (e-mail: {zhangzij15@tsinghua.org.cn; shenyuan\_ee@tsinghua.edu.cn}). {\it (Corresponding author: Yuan Shen.)}
       }
}

\maketitle
\markboth{}{ZHANG AND SHEN: CONCURRENT CODED SIGNAL-MULTIPLEXING RANGING}
\begin{abstract}
Signal-multiplexing network ranging (SM-NR) shares broadcasts across node pairs, but its sequential operation leads to a ranging cycle that grows linearly with network size. This paper proposes a concurrent coded SM-NR (CC-SM-NR) framework for asynchronous half-duplex networks. 
Firstly, the CC-SM-NR protocol coordinates concurrent transmissions through binary transmit--listen codewords. The transmit--listen schedule defined by these codewords ensures reciprocal observations subject to a finite concurrency limit. Then, we derive the exact minimum number of transmit--listen rounds without a concurrency limit, which reveals that the minimum grows logarithmically with network size. To account for practical scenarios, we establish the necessary and sufficient conditions for the constant-weight feasibility of codewords under a finite concurrency limit. Subsequently, we propose a low-complexity scheduling algorithm that achieves the minimum round count within the constant-weight codeword class. To support higher observation redundancy, this scheduling design is extended through a greedy construction. Finally, simulation results demonstrate the effectiveness of the proposed schemes for network ranging.
\end{abstract}
	
\acresetall	
\section{Introduction}\label{sec:intro}
Network localization relies on inter-node distance measurements to determine the relative geometry of a network \cite{WinSheDai:J18,zhang2022smnr}. Recent studies have developed scalable inference methods and distributed estimation methods \cite{teague2022scalable,li2022convergence,Li2026ScalableCooperative}, which achieve considerable localization performance in large cooperative networks \cite{xiong2022massive}. The resulting accuracy heavily depends on the quality of range measurements \cite{SheWymWin:J10,gomezvega2024deployment}. Hence, acquiring accurate and timely range measurements is an important prerequisite for network localization.

Range acquisition becomes challenging when nodes operate with independent clocks. The clock offsets and frequency deviations place the recorded timestamps on different time scales. When the signal-exchange interval is much longer than the time of flight (ToF), even a small frequency mismatch can produce a substantial ranging error \cite{neirynck2016alternative,Shalaby2024Timing}. The clock-scale compensation is thus essential for reliable ranging and the resulting localization. Another important consideration is the time required to acquire a complete set of pairwise distances. As the network grows, the sequential signal exchange may lengthen the ranging cycle. Within a given time budget, the number of complete measurements is therefore limited, which finally limits the localization accuracy.

Fortunately, concurrent transmissions can shorten the range acquisition time, provided that the node interactions can support valid range measurements. In a half-duplex network, nodes transmitting in the same round cannot receive one another's signals. Their reciprocal observations should therefore be obtained in other rounds. The available orthogonal frequency or code resources, together with the receiver capability, also limit the number of simultaneous transmitters \cite{russell2015neighbor,ellison2022multinode}. These constraints determine that the schedule of transmit and listen states is crucial to efficient concurrent network ranging.

\subsection{Related Work}
In range estimation, the clock synchronization and propagation-delay estimation are closely coupled, with their identifiability depending on the assumptions on message delays \cite{FrerisGrahamKumar2011Limits}. The literature has revealed that two-way timing can support joint clock-offset and skew estimation \cite{NohChaudhariSerpedinSuter2007}, as well as joint ranging and clock-parameter estimation \cite{DwiAngZacHan:J15}. The double-sided ranging reduces the effect of clock-frequency deviations \cite{neirynck2016alternative}, while asymmetrical time-stamping and passive listening exploit shared observations for joint synchronization and ranging \cite{CheRajLeuvan:J13}. The joint position and clock estimation has also been studied through mobile-network models \cite{rajan2015joint} and clock-rigidity analysis \cite{WenSchoofChapman2023Clock}. More recent work addresses clock tracking from ordinary network traffic \cite{WangLuPengLi2021Tracking}, phase and frequency synchronization \cite{ZinoDaboraPoor2025HD}, and ranging for dynamic nodes \cite{SunWangShenHuang2025TRT}. Recent work also addresses robust clock-parameter estimation under unknown packet-delay distributions and delay asymmetry \cite{Wang2026CompetitiveClock}. Regardless of the range estimator, when each pairwise range is acquired through node-to-node signaling, the number of signals grows quadratically with network size.

The broadcast-based methods reduce this overhead by allowing multiple measurement relationships to share one transmission. The pairwise broadcast synchronization reuses overheard timing exchanges \cite{NohSerQar:J08}, and recent multi-link overhearing methods further exploit this principle for clock-parameter estimation \cite{LiuWang2025Overhearing}. For ranging, signal-multiplexing network ranging (SM-NR) acquires all pairwise ranges in an $N$-node network using $N+1$ sequential rounds \cite{zhang2022smnr}. Its first and last reference broadcasts provide a common interval for clock-scale compensation. This mechanism has also been implemented in a multi-robot localization system \cite{zhao2025stereo}. Other broadcast protocols reuse timestamps across successive frames in densely meshed networks \cite{gentner2023dense} or support swarm ranging in dynamic networks \cite{shan2022swarm}. A recent swarm-ranging protocol maximizes the number of distance calculations, which are obtained from a set of signaling events \cite{hou2026optimal}. The signal-multiplexing network measuring further addresses the timeliness of asynchronous vehicle localization \cite{zhao2024timeliness}. These methods demonstrate how shared transmissions and timestamps improve network measurement efficiency. However, due to their sequential single-node transmission, the complete ranging cycle is still proportional to network size.

Fortunately, concurrent reception offers an additional means of reducing the occupied rounds. 
By utilizing orthogonal reference signal resources, multiple arrival times can be extracted from overlapping responses \cite{corbalan2020concurrent,schuh2024complex}. In distributed arrays, the distinguishable time--frequency waveforms have also been used for simultaneous inter-node ranging \cite{ellison2022multinode}. Cooperative localization has combined multi-node ranging with physical clock synchronization using code-division resources \cite{gu2024cooperative}. In multicarrier sensing, spreading mitigates cross-user ranging interference caused by target delays beyond the cyclic prefix \cite{Said2026Spreading}. These physical-layer and synchronization mechanisms support simultaneous signal acquisition, whereas a half-duplex network still requires transmit--listen arrangements that provide observations in both directions. Binary transmit--receive patterns have been investigated for energy-efficient link assessment \cite{keshavarzian2004link}, and neighbor discovery with finite multipacket reception already includes half-duplex operation and network-wide reception coverage \cite{russell2015neighbor}. In practical systems, the ranging problem additionally requires clock-compensated reciprocal measurements, a pre-set number of observations in each direction between non-reference nodes, and a ranging cycle whose length can be minimized under the concurrency limit.

Measurement scheduling and resource allocation provide a complementary solution. In existing works, the joint power and bandwidth allocation has been studied for cooperative localization with asynchronous request--reply broadcasts \cite{ZhaMolSheZhaFenWin:J16}. The navigation scheduling has been analyzed \cite{WanSheConWin:J17} and addressed using distributed algorithms \cite{wang2019distributed}. The learning-based scheduling can reduce the measurement count \cite{Peng2019Scheduling} and select relative measurements under communication constraints \cite{zhu2022scheduling}. The joint sensor and communication-link scheduling further reduces communication cost subject to localization requirements \cite{Wang2025LinkScheduling}. The communication-aware vehicle selection also accounts for interference and localization geometry \cite{Li2025VehicleSelection}. The duty-cycled ranging protocols in \cite{vecchia2025sonar} further demonstrate the tradeoff among update rate, discovery latency, and power consumption. These studies aim to efficiently exploit measurement resources. Differently, our scheduling objective is to obtain complete reciprocal coverage within a short cycle while controlling the transmission count of network nodes. 

Consequently, an important problem is to determine {\it how many rounds are required for this task} and {\it how to construct the corresponding schedules efficiently}. This problem also involves an accuracy--time tradeoff. For two-way ranging, response-delay optimization balances individual measurement variance against the information acquired per unit time \cite{Shalaby2024Timing}. Additional observations can provide more measurements but increase the transmission burden, which may lengthen the cycle. Conversely, packing more transmissions into each round reduces the occupied time but can increase timestamp measurement error \cite{schuh2024complex}. The number of rounds, observation redundancy, and measurement accuracy should therefore be evaluated together under the intended time budget.

\subsection{Our Contributions}
To accelerate network ranging, this paper proposes the concurrent coded SM-NR (CC-SM-NR) framework for asynchronous half-duplex networks with resolvable concurrent signals. The basic idea is to assign binary codewords that determine the nodes' transmit and listen states across rounds. With carefully designed codewords, a transmission-balanced schedule with a minimum number of transmit--listen rounds can be obtained. The main contributions are summarized as follows.

\begin{itemize}
\item We establish the CC-SM-NR framework and its signal-interaction protocol. Under half-duplex operation, transmit--listen rounds support concurrent transmissions by multiple nodes while preserving bidirectional transmit--listen opportunities. The protocol integrates the timestamp observations with clock-scale compensation and ToF estimation. It also supports a finite concurrency limit and a prescribed directional observation requirement.
\item We derive the minimum number of transmit--listen rounds. Without a concurrency limit, we obtain the exact minimum and show its logarithmic growth with network size. Under a finite concurrency limit, we establish necessary and sufficient feasibility conditions for constant-weight codeword schedules, which determines the minimum round count. These analyses can quantify how concurrent transmissions influence the ranging-cycle duration.
\item We develop a low-complexity scheduling algorithm that combines parameter search with load-balancing exchanges to attain the minimum transmit--listen round count within the constant-weight codeword class. To improve the robustness, we extend this design through a greedy construction to support higher observation redundancy.
\item We conduct numerical simulations to evaluate the ranging and relative-localization performance. The results show that the scheduling algorithm and its extension can efficiently construct feasible schedules. Compared to the benchmark schemes, thanks to the shorter ranging cycles of the CC-SM-NR scheme, the proposed schemes can improve the accuracy efficiently.
\end{itemize}


\subsection{Organization and Notation}
The remainder of this paper is organized as follows. In Section~\ref{sec:sys}, we define the network-ranging problem and the clock and timestamp models. In Section~\ref{sec:protocol}, we present the CC-SM-NR protocol, clock-scale compensation, and range estimation. In Section~\ref{sec:rounds}, the minimum number of transmit--listen rounds is analyzed. In Section~\ref{sec:schedule_construction}, the scheduling algorithm and its extension are developed. In Section~\ref{sec:nsr}, we evaluate construction time, ranging efficiency, and ranging and relative-localization accuracy. In Section~\ref{sec:con}, conclusions are drawn.

\textit{Notation:} Scalars, vectors, and matrices are denoted by italic $x$, bold lowercase $\bf x$, and bold uppercase $\bf X$ symbols, respectively. Sans-serif symbols $\RM X$ denote random quantities, and a hat denotes an estimate $\hat {\RM X}$. The sets of real numbers and positive integers are $\mathbb R$ and $\mathbb N$, respectively, and $[L]=\{1,\ldots,L\}$. Superscript $(\cdot)^{\rm T}$ denotes transpose, $\|\cdot\|$ the Euclidean norm, and $\supp(\bf Z)$ the set of indices of nonzero entries in $\bf Z$. For sets and families, $|\mathcal A|$, $\mathcal A/\mathcal B$, and $\mathcal A\subseteq\mathcal B$ denote cardinality, set difference, and inclusion. The operators $d_H(\cdot,\cdot)$, $\lfloor\cdot\rfloor$, and $\lceil\cdot\rceil$ denote Hamming distance, floor, and ceiling, respectively, and $\binom Lw$ is a binomial coefficient.

\section{Preliminaries}\label{sec:sys}
In this section, the system model and the ranging task for asynchronous networks are introduced.
\subsection{Nodes and Range Parameters}
Consider an asynchronous network consisting of $N\geq3$ half-duplex nodes, which are indexed by
\begin{equation}
	\mathcal N \triangleq \{1,2,\ldots,N\}.
\end{equation}
The two-dimensional position of node $n\in\mathcal N$ is denoted by
$\bm p_n=[x_n,y_n]^{\rm T}\in\mathbb R^2$. 
Assume that the scheduled network is fully connected with line-of-sight (LOS) propagation \cite{zhang2022smnr,zhao2024timeliness}. For any two nodes
$n_1,n_2\in\mathcal N$, their pairwise range and the corresponding time of flight (ToF) are
\begin{equation}
	d_{n_1,n_2}
	=
	\|
	\bm p_{n_1}-\bm p_{n_2}
	\| {~~\text{and}~~} \operatorname{ToF}(n_1,n_2)
	=
	\frac{d_{n_1,n_2}}{c},
\end{equation}
where $c$ denotes the propagation speed of the ranging signal.

Network ranging aims to estimate the
$\binom{N}{2}$ pairwise ranges
$d_{n_1,n_2}$ for all $n_1,n_2\in{\cal N}$ from the timestamps collected during a ranging cycle. Since the nodes operate with independent local clocks, these timestamps are generally measured on different time scales and with different offsets. Each node operates in half-duplex mode and hence cannot receive while transmitting. Therefore, a network-ranging schedule should provide reciprocal transmit--listen opportunities while keeping the time required to cover all node pairs short.
In  Fig.~\ref{fig:network_scene}, we show a one-round snapshot of the existing SM-NR and the proposed CC-SM-NR. In CC-SM-NR, multiple nodes can transmit concurrently, while the remaining nodes listen.

\begin{figure}[!t]
    \centering
    \includegraphics[width=\columnwidth]{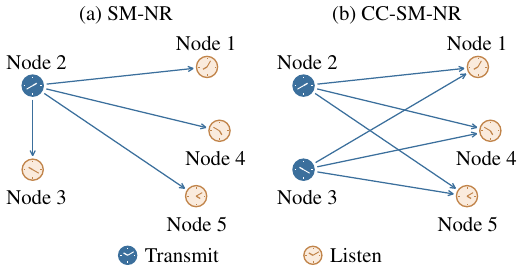}
    \caption{A one-round snapshot of a half-duplex network: (a) SM-NR with one transmitter and (b) CC-SM-NR with two concurrent transmitters.}
    \label{fig:network_scene}
\end{figure}

\subsection{Clock and Timestamp Model}
To characterize the asynchronous clocks of nodes, we introduce the well-known timestamp model \cite{zhang2022smnr,CheRajLeuvan:J13,neirynck2016alternative,zhao2024timeliness,WenSchoofChapman2023Clock}. Let $t_\alpha^{(n)}$ denote the true time at which node $n$ transmits or receives a signal in a ranging event $\alpha$. The timestamp observed at node $n$ is
modeled as
\begin{equation}
	{\rv \gamma}_\alpha^{(n)}
	={\rv{k}}^{(n)}t_\alpha^{(n)}+\theta^{(n)}+{\rv{w}}_\alpha^{(n)},
	\label{eq:clock_model}
\end{equation}
where $\theta^{(n)}$ is the clock offset, ${\rv{k}}^{(n)}=1+\rv{e}^{(n)}$ is the clock skew, $\rv{e}^{(n)}$ is the clock-frequency deviation, and
$\rv{w}_\alpha^{(n)}$ is the timestamp measurement error. Following \cite{zhang2022smnr,CheRajLeuvan:J13,WenSchoofChapman2023Clock,neirynck2016alternative,zhao2024timeliness}, we temporarily neglect the measurement errors $\rv{w}_\alpha^{(n)}$ in the algorithmic derivations and evaluate their effect through numerical simulations.

Following the clock model defined in the IEEE 802.15.4a standard \cite{IEEEstd2}, the clock-frequency deviation satisfies ${\rv{e}}^{(n)}\in[-e_{\max},e_{\max}]$. Since the cycle is short relative to the timescale of oscillator variation, the clock skews are assumed constant within each ranging cycle. For any two events $\mu$ and $\nu$ observed by node $n$, the true time interval and the observed timestamp interval are respectively defined  as
\begin{align}
	{D}_{\mu,\nu}^{(n)}
	\triangleq
	{t}_{\nu}^{(n)}-{ t}_{\mu}^{(n)}
	~~{\text{and}}~~
	\msf{D}_{\mu,\nu}^{(n)}
	\triangleq
	{\rv \gamma}_{\nu}^{(n)}-{\rv \gamma}_{\mu}^{(n)}.
	\label{eq:D_definition}
\end{align}
Taking timestamp differences removes the $N$ clock offsets $\theta^{(n)}$ in the $N$-node system, while the clock skew ${\rv{k}}^{(n)}$ should be carefully addressed in range estimation. In fact, accurate clock-scale compensation is essential, because the intervals between events can be much longer than the ToF. In uncompensated ranging, even a small relative frequency deviation produces a ToF estimation error proportional to the time interval among ranging events \cite{neirynck2016alternative,zhang2022smnr,WinSheDai:J18,zhao2024timeliness}.

\section{Concurrent Coded SM-NR Framework}\label{sec:protocol}
The proposed CC-SM-NR framework is presented in this
section. Its ranging protocol, synchronization interval, transmit--listen codewords, and range estimator are illustrated in turn. In Section~\ref{subsec:ranging_protocol}, the signal-interaction protocol is described. In Section~\ref{subsec:synchronization_interval}, the clock-scale compensation is explained. In Section~\ref{subsec:transmit_listen_codewords}, we define the codewords and bidirectional observation opportunities for network ranging. In Section~\ref{subsec:range_estimation}, the range estimator is derived from timestamps.

\subsection{Ranging Protocol}
\label{subsec:ranging_protocol}

The existing SM-NR scheme estimates the full-network range information through $N+1$ sequential node-by-node rounds \cite{zhang2022smnr}. Since its cycle duration increases linearly with $N$, the update frequency of
localization information is limited. As a solution, the proposed CC-SM-NR framework reduces the ranging-cycle duration by
allowing several nodes to transmit in a single round, while assigning their transmit and listen states to achieve reciprocal observations. The protocol of CC-SM-NR consists of three stages:
\begin{enumerate}
	\item \emph{Initial reference stage:} Node 1, i.e.,  the reference node, broadcasts a ranging
	signal and all other nodes listen.
	\item \emph{Concurrent coded stage:} The non-reference nodes from ${\cal N}/\{1\}$ execute
	$R$ {\it transmit--listen rounds} according to the binary transmit--listen codewords. The transmit-state nodes  broadcast signals, while listen-state nodes receive.
	\item \emph{Final reference stage:} Node 1 broadcasts the reference
	ranging signal again and all other nodes listen.
\end{enumerate}
The complete ranging cycle therefore occupies $N_{\rm round}=R+2$ rounds in total, which can be smaller than the $N+1$ rounds required by SM-NR when the network size is large. For example, Fig.~\ref{fig:protocol_schedule} illustrates a 7-node protocol, in which the CC-SM-NR framework preserves reciprocal timestamp observations while reducing the complete ranging cycle.

\begin{figure}[!t]
	\centering
	\includegraphics[width=\columnwidth]{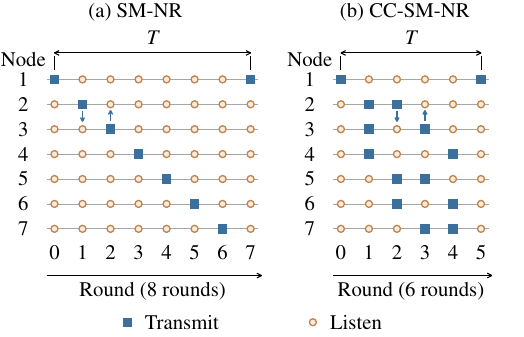}
	\caption{The protocols for a 7-node network: (a) SM-NR occupies $N+1=8$ rounds and (b) CC-SM-NR occupies $R+2=6$ rounds. The vertical arrows highlight reciprocal observations between nodes 2 and 3.}
	\label{fig:protocol_schedule}
\end{figure}

We define the set of simultaneous transmitters in transmit--listen round $r$ as
\begin{equation}
	\mathcal S_r\triangleq
	\{n\in\mathcal N/\{1\}:\chi_{n,r}=1\},\quad \forall r\in[R].
\end{equation}
To describe protocol operation above the physical layer, we introduce the following definition. 
\begin{definition}[Resolvable concurrent signals]
	In round $r$, every listening node can identify each transmitter
	$n\in\mathcal S_r$ and produce a timestamp described by
	\eqref{eq:clock_model}. Waveform-dependent estimation uncertainty, including errors associated with code-division or frequency-division multiplexing, is included in the timestamp measurement error ${\rv{w}}_\alpha^{(n)}$. 
\end{definition}
\begin{remark}
With orthogonal resource allocation, the resolvable concurrent signals are feasible in practice \cite{corbalan2020concurrent,schuh2024complex,ellison2022multinode,gu2024cooperative,Said2026Spreading}. For example, experiments have demonstrated multiple response arrival times extracted from overlapping signals using response-dependent time shifts \cite{corbalan2020concurrent}. It is proved that, the ranging error scales as $\sqrt{K_r}$ for time-division or code-division multi-user systems, where $K_r=|\mathcal S_r|$. Instead, in frequency-division case, ranging error has a $K_r$-scaling law \cite{SheWin:J10a}. Our simulations in Section \ref{sec:nsr} focus on the former case. 
\end{remark}

To take the ability to resolve concurrent signals, as well as the additional timestamp measurement error induced by concurrent transmissions, into account, we introduce a prescribed concurrency limit and timestamp-error model in later sections.

\subsection{Synchronization Interval}
\label{subsec:synchronization_interval}
In the SM-NR protocol, as a necessary condition, node 1 broadcasts reference signals at the beginning and end of a ranging cycle to provide a common interval for clock-scale compensation. For a fixed cycle duration, placing the
reference broadcasts at the beginning and end provides the longest common observation interval. It helps to minimize the contribution of timestamp measurement error to clock-scale compensation within the cycle \cite{zhang2022smnr}. To maintain the advantage, the proposed CC-SM-NR retains this design. Let $T$ be
the true interval between the two transmissions of node 1, referred to as the synchronization interval. Because the
propagation delay is the same for both transmissions, the interval $T$ is observed at every node $n$. Thus, $T$ measured by node $n$ is 
\begin{equation}
	\msf{T}^{(n)}={\rv{k}}^{(n)}T.
	\label{eq:T_observed}
\end{equation}
Then, given the distribution
$h_{{\rv{e}}^{(n)}}(\cdot)$ of the clock-frequency deviation ${\rv{e}}^{(n)}$,
the synchronization interval can be estimated by maximum likelihood (ML):
\begin{equation}
	\widehat{\msf{T}}
	=\argmax_T\prod_{n\in\mathcal N}
	h_{{\rv{e}}^{(n)}}\!\left(\msf{T}^{(n)};T\right),
	\label{eq:T_ML_generic}
\end{equation}
where the density is expressed in likelihood form. For
the uniform distribution ${\rv{e}}^{(n)}\sim\mathcal U[-e_{\max},e_{\max}]$, the ML-based estimator of $T$ is derived in \cite{zhang2022smnr}, written as
\begin{equation}
	\widehat{\msf{T}} = 
	\frac{\max_{n\in\mathcal N}\msf{T}^{(n)} }{1+e_{\max}}.
	\label{eq:T_uniform}
\end{equation} 
The estimate $\widehat{\msf T}$ is used to compensate for the node-dependent clock scales in range estimations.
\begin{remark}
	The synchronization interval  $\widehat{\msf{T}}$ is reused in the clock-scale compensation of all range estimates, and its error actually dominates in the factors causing ranging errors \cite[Sec.~IV-C]{zhang2022smnr}. For the proposed CC-SM-NR, another advantage of allowing only one node to transmit at each reference event is that it can obtain a sufficiently precise estimate $\widehat{\msf{T}}$, which is not influenced by the additional timestamp measurement errors induced by concurrent transmissions. As a result, concurrent transmissions are confined to the $R$ transmit--listen rounds to protect the performance of clock-scale compensation.
\end{remark}

\subsection{Transmit--Listen Codewords}
\label{subsec:transmit_listen_codewords}
For each node $n\in\mathcal N$, we define a codeword that specifies its transmit or listen state in the concurrent coded stage:
\begin{equation}
	\bm\chi_n=[\chi_{n,1},\ldots,\chi_{n,R}]^{\rm T}\in\{0,1\}^{R},
\end{equation}
where $\chi_{n,r}=1$ indicates that node $n$ transmits and $\chi_{n,r}=0$ indicates that it listens in round $r$. In particular, we fix ${\bm\chi}_1=[0,\cdots,0]^{\rm T}$, ${\chi}_{1,0}={\chi}_{1,R+1}=1$ and ${\chi}_{n,0}={\chi}_{n,R+1}=0$ for $n\ne1$ to describe the signal-interaction events between node 1 and other nodes. Collecting the codewords of the $N-1$ non-reference nodes gives the schedule matrix
\begin{equation}
	\bm X=[\bm\chi_2,\cdots,\bm\chi_N]
	\in\{0,1\}^{R\times {(N-1)}},
\end{equation}
whose rows represent transmit--listen rounds and columns represent non-reference nodes. For the node pair $(n_1,n_2)$, the number of directional transmit--listen
opportunities from $n_1$ to $n_2$ is defined as
\begin{equation}
	\Gamma_{n_1\rightarrow n_2}(\bm X)
	\triangleq
	\sum_{r=1}^{R}\chi_{n_1,r}(1-\chi_{n_2,r}).
	\label{eq:gamma_direction}
\end{equation}
$\Gamma_{n_1\rightarrow n_2}(\bm X)$ counts the rounds in which node $n_1$ transmits and $n_2$ listens among the $R$ transmit--listen rounds. To enable ranging for every node pair, the schedule $\bf X$ should satisfy the following {\it bidirectional feasibility condition}.
\begin{definition}[Bidirectional feasibility condition]
	A schedule $\bm X$ is bidirectionally feasible if
	\begin{equation}
		\Gamma_{n_1 \rightarrow n_2}(\bm X)\geq 1,
		\quad \forall n_1\neq n_2,
		\quad n_1,n_2\in\mathcal N/\{1\}.
		\label{eq:bi_feasible}
	\end{equation}
\end{definition}
This condition excludes node 1 because the two reference broadcasts transmitted by node 1 provide observations from node 1 to any node $n$, and each non-reference node provides an outgoing observation to node 1. 
Thus, condition \eqref{eq:bi_feasible} can provide bidirectional coverage of every node pair.

\begin{table}[t]
	\centering
	\caption{A Transmit--Listen Schedule $\bm X^{\rm T}$ for $N=10$ and $R=5$.}
	\label{tab:schedule_example}
	\begin{tabular}{c*{5}{c}}
		\toprule
		Node & $r=1$ & $r=2$ & $r=3$ & $r=4$ & $r=5$\\
		\midrule
		2  & 1&1&0&0&0\\
		3  & 1&0&1&0&0\\
		4  & 1&0&0&1&0\\
		5  & 1&0&0&0&1\\
		6  & 0&1&1&0&0\\
		7  & 0&1&0&1&0\\
		8  & 0&1&0&0&1\\
		9  & 0&0&1&1&0\\
		10 & 0&0&1&0&1\\
		\midrule
		\shortstack{Concurrent\\transmissions} &4&4&4&3&3\\
		\bottomrule
	\end{tabular}
\end{table}

\begin{remark}
	Table~\ref{tab:schedule_example} displays the schedule matrix ${\bm X}^{\rm T}$ for a bidirectionally feasible schedule with five transmit--listen rounds with $N=10$. The rows specify the codewords of the nine non-reference nodes. In each column, at most four nodes transmit simultaneously. For example, nodes 2 and 3 have codewords $\bm\chi_2=[1,1,0,0,0]^{\rm T}$ and $\bm\chi_3=[1,0,1,0,0]^{\rm T}$. Round $r=2$ provides a node-$2$-to-node-$3$ observation, whereas round $r=3$ provides a
	node-$3$-to-node-$2$ observation. Including the two reference broadcasts, the complete cycle occupies $R + 2 = 7$ rounds, compared with the $N+1=11$ rounds required by the SM-NR scheme.
\end{remark}

\subsection{Range Estimation}
\label{subsec:range_estimation}

Consider two nodes $n_1,n_2\in\mathcal N$. Under the bidirectional feasibility condition in (\ref{eq:bi_feasible}), there exist two reciprocal events $\mu=(n_1,r_{n_1})$ and $\nu=(n_2,r_{n_2})$ with the following properties: 
\begin{enumerate}
	\item In event $\mu$, node $n_1$ transmits while node $n_2$
	listens during $\mu$, i.e., $\chi_{n_1,r_{n_1}}=1$ and $\chi_{n_2,r_{n_1}}=0$. 
	\item In event $\nu$, node $n_2$ transmits while node $n_1$ listens during $\nu$, i.e., $\chi_{n_2,r_{n_2}}=1$ and $\chi_{n_1,r_{n_2}}=0$. 
\end{enumerate}
Without loss of generality, assume that $\mu$ precedes $\nu$. Node $n_1$ timestamps its transmission in event $\mu$ and its subsequent reception from node $n_2$ in event $\nu$. Meanwhile, node $n_2$ timestamps the reception from node $n_1$, and node $n_2$ then timestamps its transmission in event $\nu$. The true event times therefore satisfy
\begin{align}
	t_{\mu}^{(n_2)} = t_{\mu}^{(n_1)} + \ToF(n_1,n_2),\\
	t_{\nu}^{(n_1)} = t_{\nu}^{(n_2)} + \ToF(n_1,n_2).
\end{align}
Using the time intervals in \eqref{eq:D_definition} gives
\begin{align}
	{\msf D}_{\mu,\nu}^{(n_1)}
	&={\rv \gamma}_{\nu}^{(n_1)}-{\rv \gamma}_{\mu}^{(n_1)}={\rv{k}}^{(n_1)}(t_\nu^{(n_1)}-t_\mu^{(n_1)}),
	\label{eq:D_n1}\\
	{\msf D}_{\mu,\nu}^{(n_2)}
	&={\rv \gamma}_{\nu}^{(n_2)}-{\rv \gamma}_{\mu}^{(n_2)}={\rv{k}}^{(n_2)}(t_\nu^{(n_2)}-t_\mu^{(n_2)}).
	\label{eq:D_n2}
\end{align}
Recalling that $\msf{T}^{(n)}={\rv{k}}^{(n)}T$, the true ToF is
\begin{equation}
	\ToF(n_1,n_2)
	=\frac{T}{2}
	\left(
	\frac{{\msf D}_{\mu,\nu}^{(n_1)}}{{\msf T}^{(n_1)}}
	-
	\frac{{\msf D}_{\mu,\nu}^{(n_2)}}{{\msf T}^{(n_2)}}
	\right).
	\label{eq:ToF_identity}
\end{equation}
Using the invariance property of ML estimation, the ML-based ToF estimator is
\begin{equation}
	\widehat\ToF(n_1,n_2)
	=\frac{\widehat{\msf{T}}}{2}
	\left(
	\frac{{\msf D}_{\mu,\nu}^{(n_1)}}{{\msf T}^{(n_1)}}
	-
	\frac{{\msf D}_{\mu,\nu}^{(n_2)}}{{\msf T}^{(n_2)}}
	\right).
	\label{eq:ToF_estimator}
\end{equation}
The corresponding range estimate is $\widehat d_{n_1,n_2}=c\,\widehat{\ToF}(n_1,n_2)$.

If multiple reciprocal events are available for measuring $d_{n_1,n_2}$, let $\mathcal P_{n_1,n_2}$ denote
the set of reciprocal event pairs, and let $\widehat{d}_{n_1,n_2}(\mu,\nu)$ denote the range estimate for event pair $(\mu,\nu)\in\mathcal P_{n_1,n_2}$. In this case, the combined range estimate can be written as
\begin{equation}
	\widehat{d}_{n_1,n_2}
	=\sum_{(\mu,\nu)\in\mathcal P_{n_1,n_2}}
	\beta_{\mu,\nu}^{(n_1,n_2)}
	\widehat{d}_{n_1,n_2}(\mu,\nu),
	\label{eq:weighted_ToF}
\end{equation}
where $\beta_{\mu,\nu}^{(n_1,n_2)}\geq0$ and the weights sum to one. For simplicity, an equal-weight setup for $\beta_{\mu,\nu}^{(n_1,n_2)}$ can be adopted in practice \cite{zhang2022smnr}.

\section{Analysis of Transmit--Listen Round Counts}\label{sec:rounds}
In this section, we derive the minimum number of transmit--listen rounds for the proposed CC-SM-NR framework by utilizing strongly separating theory \cite{bollobas2007separating}. In Section~\ref{subsec:unrestricted_rounds}, we give the minimum round count and its scaling without a concurrency limit. In Section~\ref{subsec:constrained_rounds}, the feasibility of constant-weight codewords is established, and the minimum round count under a finite concurrency limit is derived.

\subsection{Transmit--Listen Rounds Without Concurrency Limit}
\label{subsec:unrestricted_rounds}
Let ${\cal A}_n \triangleq \supp({\bm\chi_n})$ denote the set of nonzero positions in codeword ${\bm\chi_n}$. The subsets of transmit--listen rounds assigned to the $N-1$ non-reference nodes form the codeword family
\begin{equation}
	\mathcal F\triangleq
	\left\{ {\cal A}_n:n\in\mathcal N/\{1\} \right\},
\end{equation}
which corresponds to a schedule matrix $\bm X$. Using the following lemma, one can establish the relation between
the bidirectional feasibility and this family.
\begin{lemma}[Bidirectional feasibility]
	\label{lem:antichain}
	The transmit--listen schedule $\bm X$ satisfies the bidirectional feasibility condition among the
	non-reference nodes if and only if $
	|\mathcal F|= N-1
	$ and $\mathcal F$ is an antichain. Equivalently, for every pair of	distinct nodes $n_1,n_2\in\mathcal N/\{1\}$, we have
	$
\supp(\bm\chi_{n_1})\nsubseteq\supp(\bm\chi_{n_2})
$ and $
\supp(\bm\chi_{n_2})\nsubseteq\supp(\bm\chi_{n_1}).$
\end{lemma}

\begin{IEEEproof}
	Let
	$
	{\cal A}_{n_1}=\supp(\bm\chi_{n_1})$ and $
	{\cal A}_{n_2}=\supp(\bm\chi_{n_2})$.
	The number of the transmit--listen rounds in which node $n_1$ transmits while
	node $n_2$ listens is
	\begin{equation}
	\Gamma_{{n_1}\rightarrow {n_2}}
	=
	\sum_{r=1}^{R}\chi_{n_1,r}(1-\chi_{{n_2},r})
	=
	|{\cal A}_{n_1}/ {\cal A}_{n_2}|.
	\end{equation}
	Consequently,
	$
	\Gamma_{{n_1}\rightarrow {n_2}}\geq1$ implies that ${\cal A}_{n_1}/{\cal A}_{n_2}\neq\varnothing$, which leads to ${\cal A}_{n_1}\nsubseteq {\cal A}_{n_2}$.
	Similarly, we have
	$
	{\cal A}_{n_2}\nsubseteq {\cal A}_{n_1}
	$.
	Hence, bidirectional feasibility holds if and only if the non-reference nodes' codeword supports are pairwise incomparable. Besides, pairwise incomparability also avoids
	identical subsets, so the assigned subsets are distinct. One can obtain that $|\mathcal F|=N-1$ and $\mathcal F$ is an antichain. The converse follows immediately.
\end{IEEEproof}
Applying the strongly separating-system theory \cite[Theorem~2]{bollobas2007separating} to Lemma~\ref{lem:antichain}, we derive the minimum value of $R$ in the following theorem.
\begin{theorem}[Minimum value of $R$]
	\label{thm:sperner}
	For an $N$-node network, the minimum number of transmit--listen rounds is
	\begin{equation}
		R^\star_\infty
		=\min\left\{R\in\mathbb N:
		\binom{R}{\lfloor R/2\rfloor}\geq N-1\right\}.
		\label{eq:R_opt_inf}
	\end{equation}
\end{theorem}

\begin{IEEEproof}
	According to Lemma~\ref{lem:antichain}, the $N-1$ codeword supports must form an
	antichain in the Boolean lattice $2^{[R]}$. Sperner's theorem states that an
	antichain in $2^{[R]}$ contains at most
	$\binom{R}{\lfloor R/2\rfloor}$ members \cite{sperner1928satz}. Therefore, any
	feasible schedule must satisfy the inequality in \eqref{eq:R_opt_inf}.
	Conversely, every subset of $[R]$ having cardinality $\lfloor R/2\rfloor$
	belongs to the central layer of the Boolean lattice, and any two distinct
	members of this layer are incomparable. Selecting any $N-1$ such subsets yields
	a feasible constant-weight schedule. Hence, the lower bound of $R$ is achievable.
\end{IEEEproof}

\begin{corollary}[Scaling law of $R$]
	As $N\rightarrow\infty$, the unconstrained minimum of $R$ satisfies
	\begin{equation}
		R^\star_\infty
		= \log_2(N-1)
		+\frac{1}{2}\log_2\log_2(N-1)+{\cal O}(1).
		\label{eq:asymptotic_round}
	\end{equation}
\end{corollary}

\begin{IEEEproof}
Let $\Theta(\cdot)$ denote an asymptotically tight bound on the growth rate of a function. The result can be obtained by
substituting the central-binomial-coefficient approximation
$
	\binom{R}{\lfloor R/2\rfloor}
		=\Theta\!\left(\frac{2^R}{\sqrt R}\right)
$
	into \eqref{eq:R_opt_inf} and then solving for $R$.
\end{IEEEproof}
To see the round reduction, Table~\ref{tab:round_scaling} compares the exact round counts from
Theorem~\ref{thm:sperner} with those of the SM-NR scheme. At $N=5$, both schemes require the same number of occupied rounds, whereas the reduction becomes substantial for large
networks. At $N=1000$, SM-NR requires $N+1=1001$ rounds and CC-SM-NR requires $R^\star_\infty+2=15$. 

\begin{table}[t]
	\centering
	\caption{Round Counts for Different Ranging Schemes}
	\label{tab:round_scaling}
	\begin{tabular}{rccc}
		\toprule
		$N$ & SM-NR ($N+1$) & $R^\star_\infty$ & CC-SM-NR ($R^\star_\infty+2$)\\
		\midrule
		5    & 6    & 4  & 6\\
		10   & 11   & 5  & 7\\
		20   & 21   & 6  & 8\\
		50   & 51   & 8  & 10\\
		100  & 101  & 9  & 11\\
		1000 & 1001 & 13 & 15\\
		\bottomrule
	\end{tabular}
\end{table}

Theorem~\ref{thm:sperner} minimizes the number of transmit--listen
rounds, while it does not minimize the number of transmissions per node. In practical systems, a robust schedule
may require each node to transmit a limited number of times. Also, the number of concurrent transmitters in each round should be limited. For example, frequency-division multiplexing offers only a finite number of subbands for per-user ranging,  whereas code-division multiplexing provides a finite set of usable orthogonal signatures \cite{corbalan2020concurrent,schuh2024complex,ellison2022multinode,gu2024cooperative,Said2026Spreading}.
In addition, to reduce the average power consumption and ensure the users' fairness, the number of transmissions from each node in a ranging cycle should be kept as balanced as possible. The transmit--listen round count under these limits is analyzed as follows.

\subsection{Transmit--Listen Rounds Under Concurrency Limit}
\label{subsec:constrained_rounds}

Let the integer $K_{\max}\geq1$ denote the maximum number of distinguishable transmitters per round, i.e., the concurrency limit. Let the integer $q\geq1$ specify the minimum required number of observations in each direction between two non-reference nodes. The number of transmitters in round $r\in[R]$ is $K_r \triangleq \sum_{n=2}^{N}\chi_{n,r}$. Thus,
the minimum-round design problem subject to concurrency constraints is formulated as
\begin{subequations}\label{prob:min_round}
	\begin{align}
		&\underset{\substack{R\in\mathbb N\\{\bm X}\in\{0,1\}^{R\times (N-1)}}}{\operatorname{min}}\quad
		R \label{prob:min_round_obj}\\
		&~~~~~~~~~\text{s.t.}\quad
		\Gamma_{{n_1}\rightarrow {n_2}}\geq q, \qquad
		\substack{\forall n_1,n_2\in\mathcal N/\{1\}\\n_1\neq n_2}, \label{prob:dir_constraint}\\
		&~~~~~~~~~~~~~~~K_r\leq K_{\max}. \qquad
		\forall r\in[R], \label{prob:column_constraint}
	\end{align}
\end{subequations}
For $q=1$ and $K_{\max}\geq N-1$, this problem reduces to the setting of Theorem~\ref{thm:sperner}. For $q=1$ and $K_{\max}=1$, one can obtain $R=N-1$, which recovers the $N+1$ occupied rounds of the SM-NR scheme. For fixed $K_{\max}$, every non-reference node should transmit at least $q$ times, so $R\geq\lceil (N-1)q/K_{\max}\rceil$. 

To analyze the minimum $R$, we focus on the case with minimum directional observation requirement, i.e., $q=1$. To simplify notations, we define
\begin{equation}
M\triangleq N-1.
\end{equation}
Consider the practical case when each non-reference node transmits in exactly $w$ transmit--listen rounds and listens in the remaining $R-w$ rounds. In other words, every non-reference node is assigned a distinct length-$R$ binary codeword of weight $w$. In this way, the schedule ${\bm X}$ can be constructed from a constant-weight class, which provides a simple way to limit each node's transmission load. 

Importantly, constant-weight construction simplifies the bidirectional feasibility condition: For $q=1$, any collection of distinct constant-weight codewords forms an antichain and therefore satisfies the
bidirectional feasibility condition in Lemma~\ref{lem:antichain}. However, it remains to determine {\it whether these codewords can be selected to satisfy the concurrency limit.} To tackle this question, we first introduce the following lemma.
\begin{lemma}[Almost-regular constant-weight codewords]
	\label{lem:almost_regular_cw}
	For any integers $R$, $w$, and $M$ satisfying
	$
		1\leq M\leq \binom{R}{w}
	$,
	there exists a collection of $M$ distinct length-$R$ weight-$w$ codewords
	whose per-round transmitter counts satisfy
	\begin{equation}
		\max_{r\in[R]}K_r-\min_{r\in[R]}K_r\leq 1.
		\label{eq:almost_regular_load}
	\end{equation}
	When the total number of transmissions $Mw$ is divided by $R$, we obtain the following equations:
	\begin{equation}
		u=\left\lfloor\frac{Mw}{R}\right\rfloor,\quad
		s=Mw-uR,\quad 0\leq s<R,
		\label{eq:load_division}
	\end{equation}
	where $u$ and $s$ denote the quotient and remainder, respectively.
	Since $\sum_{r=1}^{R}K_r=Mw$ and the per-round transmitter counts differ by at most one, exactly $s$ rounds satisfy $K_r=u+1$, while the remaining $R-s$ rounds satisfy $K_r=u$.
	Consequently, the minimum achievable peak number of per-round concurrent transmitters is
	\begin{equation}
		\min_{\cal G}\max_{r\in[R]}K_r
		=
		\left\lceil\frac{Mw}{R}\right\rceil,
		\label{eq:optimal_peak_load}
	\end{equation}
	where $\mathcal G$ is a collection of $M$ distinct weight-$w$ subsets of $[R]$.
\end{lemma}

\emph{Proof:} See Appendix~\ref{app:proof_almost_regular_cw}.

By specializing this lemma, one can obtain the following theorem.
\begin{theorem}[Necessary and sufficient condition for constant-weight feasibility]
	\label{thm:exact_cw_feasibility}
	For $q=1$, there exists a bidirectionally feasible length-$R$ and weight-$w$ transmit--listen
	schedule serving $M$ non-reference nodes, which satisfies
	\begin{equation}
		K_r\leq K_{\max},\qquad \forall r\in[R],
	\end{equation}
	if and only if
	\begin{align}
		\binom{R}{w}&\geq M,
		\label{eq:cw_codeword_condition}\\
		Mw&\leq RK_{\max}.
		\label{eq:cw_capacity_condition}
	\end{align}
\end{theorem}

\begin{IEEEproof}
	Condition \eqref{eq:cw_codeword_condition} is necessary because there
	are only $\binom{R}{w}$ distinct length-$R$ weight-$w$ codewords. Moreover,
	the schedule contains exactly $Mw$ ones, whereas its $R$ rounds can
	contain at most $RK_{\max}$ ones, which establishes
	\eqref{eq:cw_capacity_condition}. For sufficiency, suppose that
	\eqref{eq:cw_codeword_condition} and \eqref{eq:cw_capacity_condition} hold.
	Using Lemma~\ref{lem:almost_regular_cw}, we can select $M$
	weight-$w$ codewords such that $
	\max_{r\in[R]}K_r
	=
	\left\lceil\frac{Mw}{R}\right\rceil
	$.
	Since $K_{\max}$ is an integer, \eqref{eq:cw_capacity_condition} implies that
	$
	\left\lceil\frac{Mw}{R}\right\rceil
	\leq K_{\max}
	$.
	Thus, every round satisfies the concurrency limit. Furthermore, the
	selected codewords are distinct and have equal weight, so their supports
	are pairwise incomparable. According to Lemma~\ref{lem:antichain}, the resulting
	schedule is bidirectionally feasible, which completes the proof.
\end{IEEEproof}

It follows from Theorem~\ref{thm:exact_cw_feasibility} that the minimum number of transmit--listen rounds, within the class of constant-weight codewords, can be obtained by solving
\begin{equation}
	\begin{split}
		R_{\rm cw}^{\star}
		=
		\min
		{\Big\{}
		R\in\mathbb N:
		\exists\,w\in[R]
		\quad \text{s.t.} \\
		\binom{R}{w}\geq M,\;
		Mw\leq RK_{\max}
		{\Big \}}.
	\end{split}	
	\label{eq:exact_cw_round}
\end{equation}
A parameter search can be used to determine $R_{\rm cw}^{\star}$. Specifically, $R$ is increased from $R=1$. For every candidate
$R$, all weights $w\in[R]$ should be examined in turn. The search terminates
at the first pair $(R,w)$ that satisfies the conditions in
\eqref{eq:exact_cw_round}.

In the special case when $M=\binom{R}{w}$, we have the following corollary.
\begin{corollary}[Complete constant-weight layer]
	\label{cor:complete_cw_layer}
	When $\binom{R}{w}$ distinct codewords in the length-$R$ weight-$w$ codeword class are all used for $M$ non-reference nodes,
	then every round contains exactly
	$K_r=
	\frac{w}{R}\binom{R}{w}
	$
	transmitters.
\end{corollary}

\begin{IEEEproof}
	Fix an arbitrary round $r$. A weight-$w$ codeword contains a one in
	coordinate $r$ if and only if its remaining $w-1$ ones are selected from
	the other $R-1$ coordinates. Hence, exactly
	$\binom{R-1}{w-1}$ codewords transmit in round $r$, and the same argument
	applies to every round. Thus we have
	\begin{equation}
		K_r=\binom{R-1}{w-1}
		=
		\frac{(R-1)!}{(w-1)!(R-w)!}
		=
		\frac{w}{R}\binom{R}{w},
	\end{equation}
	which completes the proof.
\end{IEEEproof}

\section{Low-Complexity Schedule Design}
\label{sec:schedule_construction}
Building on the round count analysis, we develop an optimal constant-weight scheduling algorithm in Section~\ref{subsec:q1_construction}. To improve ranging robustness, this algorithm is then extended to support higher observation redundancy in Section~\ref{subsec:qgreater1_construction}. In Section~\ref{subsec:complexity}, we analyze the computational complexity of the proposed algorithms. In Section~\ref{subsec:partial_connectivity}, we discuss the schedule feasibility in partially connected networks.

\subsection{Optimal Constant-Weight Schedule Design for $q=1$}
\label{subsec:q1_construction}
\begin{algorithm}[!t]
	\caption{Optimal Schedule Design for $q=1$}
	\label{alg:q1_exchange}
	\begin{algorithmic}[1]
		\REQUIRE $M$ and $K_{\max}$
		\ENSURE $R$, $w$, and the codeword family $\mathcal F$
		\STATE $R\leftarrow R_\infty^\star$
		\WHILE{true}
		\STATE Find $w \leftarrow w_R$ from \eqref{eq:smallest_weight_given_R}
		\IF{$Mw\leq RK_{\max}$}
		\STATE \textbf{break}
		\ENDIF
		\STATE $R\leftarrow R+1$
		\ENDWHILE
		\STATE Select any $M$ distinct weight-$w$ codewords as $\mathcal F$
		\STATE Compute $K_r$ for all $r\in[R]$
		\WHILE{$\max_rK_r-\min_rK_r>1$}
		\STATE Choose $a\in\argmax_rK_r$ and $b\in\argmin_rK_r$
		\STATE Find ${\cal A}\in\mathcal F$ with $a\in {\cal A}$ and $b\notin {\cal A}$, such that $({\cal A}/\{a\})\cup\{b\}\notin{\mathcal F}$
		\STATE Set ${\cal A}'\leftarrow ({\cal A}/\{a\})\cup\{b\}$
		\STATE $\mathcal F\leftarrow(\mathcal F/\{{\cal A}\})\cup\{{\cal A}'\}$
		\STATE $K_a\leftarrow K_a-1$ and $K_b\leftarrow K_b+1$
		\ENDWHILE
		\STATE \textbf{return} $R$, $w$, and $\mathcal F$
	\end{algorithmic}
\end{algorithm}

To find the minimum feasible round count within the constant-weight class, the optimal schedule design is summarized in Algorithm~\ref{alg:q1_exchange}, which balances the per-round transmission loads through weight-preserving exchanges. For a fixed $R$, it is sufficient to search
$1\leq w\leq\lfloor R/2\rfloor$, since a larger weight can be replaced by its complementary weight. The smallest available weight can be obtained by solving
\begin{equation}
	w_R
	\triangleq
	\min\left\{
	w\in\left\{1,\ldots,\left\lfloor\frac{R}{2}\right\rfloor\right\}:
	\binom{R}{w}\geq M
	\right\}.
	\label{eq:smallest_weight_given_R}
\end{equation}
If $Mw_R>RK_{\max}$, no larger weight can satisfy the concurrency limit
for the same $R$. Therefore, only $w_R$ needs to be checked for each candidate $R$. The search can start from the unconstrained optimum $R_\infty^\star$ in
\eqref{eq:R_opt_inf}. Once a feasible pair $(R,w)$ is found, $M$ distinct weight-$w$ codewords can be selected. Then, according to Lemma~\ref{lem:almost_regular_cw}, the per-round loads can be balanced by the exchange procedure in Algorithm~\ref{alg:q1_exchange}.

During execution of the algorithm, Theorem~\ref{thm:exact_cw_feasibility} can ensure that the first pair found by the parameter search has $R=R_{\rm cw}^{\star}$. Moreover, Lemma~\ref{lem:almost_regular_cw} guarantees the existence of the exchange in
every iteration. Each exchange preserves codeword weight and distinctness and strictly decreases the imbalance potential $\sum_{r=1}^{R}K_r^2$. The algorithm therefore terminates
after a finite number of exchanges and returns a schedule satisfying
\begin{equation}
	\max_{r\in[R]}K_r
	=
	\left\lceil\frac{Mw}{R}\right\rceil
	\leq K_{\max}.
\end{equation}
Consequently, this search can find the optimal solution within the constant-weight class without a binary optimizer. 

\subsection{Extension to Higher Observation Redundancy}
\label{subsec:qgreater1_construction}
\begin{algorithm}[!t]
	\caption{Greedy Codeword Construction for Fixed $(R,w)$}
	\label{alg:qgreater1_greedy}
	\begin{algorithmic}[1]
		\REQUIRE $M$, $R$, $w$, $q$, and $K_{\max}$
		\ENSURE A feasible family $\mathcal F$, or failure
		\STATE $\mathcal F\leftarrow\varnothing$ and $K_r\leftarrow0$,  $\forall r\in[R]$
		\FOR{$i\in \{2,\cdots,N\}$}
		\STATE Construct $\mathcal C(\mathcal F)$ using
		\eqref{eq:admissible_candidate_set}
		\IF{$\mathcal C(\mathcal F)=\varnothing$}
		\STATE \textbf{return failure}
		\ENDIF
		\STATE Select ${\cal A}^\star$ according to \eqref{eq:greedy_load_metric}
		\STATE $\mathcal F\leftarrow\mathcal F\cup\{{\cal A}^\star\}$
		\STATE $K_r\leftarrow K_r+1$ for every $r\in {\cal A}^\star$
		\ENDFOR
		\STATE \textbf{return} $\mathcal F$
	\end{algorithmic}
\end{algorithm}

To improve the robustness of network ranging, we consider the case when each range has redundant observations, i.e., $q>1$. To extend the scheduling design to address $q>1$, for fixed $(R,w)$, we retain parameter search and replace load-balancing exchanges with the greedy construction in Algorithm~\ref{alg:qgreater1_greedy}. For two weight-$w$ supports ${\cal A}_{n_1}=\supp({\bm\chi}_{n_1})$ and
${\cal A}_{n_2}=\supp({\bm\chi}_{n_2})$ with $n_1,n_2\in{\cal N}/\{1\}$, the directional observation counts can be written as
\begin{equation}
	\Gamma_{n_1\rightarrow n_2}
	=
	\Gamma_{n_2\rightarrow n_1}
	=
	w-|{\cal A}_{n_1}\cap {\cal A}_{n_2}|
	=
	\frac{1}{2}d_H(\bm\chi_{n_1},\bm\chi_{n_2}),
	\label{eq:q_hamming_relation}
\end{equation}
where $d_H(\bm\chi_{n_1},\bm\chi_{n_2})$ denotes the Hamming distance
between codewords $\bm\chi_{n_1}$ and $\bm\chi_{n_2}$. Thus, the requirement of at least $q$ observations per direction is equivalent to
$d_H(\bm\chi_{n_1},\bm\chi_{n_2})\geq 2q$, or $|{\cal A}_{n_1}\cap {\cal A}_{n_2}|\leq w-q$. The weight search can therefore be restricted to
$q\leq w\leq R-q$. Given a family of selected codewords $\mathcal F$, the candidate set can be defined as
\begin{equation}
	\mathcal C(\mathcal F)
	=
	\left\{
	{\cal A}\subseteq [R]:
	\begin{array}{l}
		|{\cal A}|=w,\\
		|{\cal A}\cap {\cal B}|\le w-q,\quad \forall {\cal B}\in\mathcal F,\\
		K_r+1 \leq K_{\max},\quad \forall r\in {\cal A}
	\end{array}
	\right\}.
	\label{eq:admissible_candidate_set}
\end{equation}
We note that adding any candidate in $\mathcal C(\mathcal F)$ preserves both the minimum Hamming-distance requirement and the per-round concurrency limit. To balance the
per-round loads, the next codeword can be randomly selected from the minimum-load candidates, i.e.,
\begin{equation}
	{\cal A}^\star
	\in
	\argmin_{{\cal A}\in\mathcal C(\mathcal F)}
	\sum_{r\in {\cal A}}K_r.
	\label{eq:greedy_load_metric}
\end{equation}

\begin{algorithm}[!t]
	\caption{Greedy Schedule Extension for $q>1$}
	\label{alg:qgreater1_adaptive}
	\begin{algorithmic}[1]
		\REQUIRE $M$, $q$, $K_{\max}$, and $N_{\rm rst}$
		\ENSURE $R_{\rm alg}$, $w$, and a feasible family $\mathcal F$
		\STATE Compute $R_{\rm cw,LB} \leftarrow \max\{R_{\rm cw}^\star,2q\}$
		\FOR{$R\in\{R_{\rm cw,LB},\cdots,Mq-1\}$}
		\STATE Construct set ${\mathcal W}_R$ using \eqref{eq:W_R_set}
		\FOR{$w\in{\mathcal W}_R$}
		\FOR{$s \in [N_{\rm rst}]$}
		\STATE Attempt a feasible family using Algorithm~\ref{alg:qgreater1_greedy} with $(M,R,w,q,K_{\max})$
		\IF{a family $\mathcal F$ is returned}
		\STATE $R_{\rm alg}\leftarrow R$
		\STATE \textbf{return} $R_{\rm alg}$, $w$, and $\mathcal F$
		\ENDIF
		\ENDFOR
		\ENDFOR
		\ENDFOR
		\STATE $R_{\rm alg}\leftarrow Mq$, $w\leftarrow q$, and ${\cal F}\leftarrow\{{\cal A}_2,\cdots,{\cal A}_N\}$ where ${\cal A}_n
		=
		\{(n-2)q+1,\ldots,(n-1)q\}$
		\STATE \textbf{return} $R_{\rm alg}$, $w$, and $\mathcal F$
	\end{algorithmic}
\end{algorithm}

To find feasible $(R,w)$, Algorithm~\ref{alg:qgreater1_adaptive} implements the extension by utilizing the greedy construction in Algorithm~\ref{alg:qgreater1_greedy}. A greedy construction may encounter
$\mathcal C(\mathcal F)=\varnothing$ before all $M$ codewords have
been found. However, this does not imply that the
parameter tuple $(M,R,w,q,K_{\max})$ is infeasible, and the greedy construction may have led to a dead end. Thus, we use multiple random restarts and progressively increase the number
of transmit--listen rounds $R$ in parameter search.

Specifically, the parameter search over $R$ can start from
$R_{\rm cw,LB} \triangleq \max\{R_{\rm cw}^\star,2q\}$, wherein  $R_{\rm cw}^\star$ is from \eqref{eq:exact_cw_round}. For a given $R$, the set of candidate weights is defined as
\begin{equation}\label{eq:W_R_set}
	{\mathcal W}_R
	=
	\left\{
	w\in
	\left\{
	q,\ldots,
	R-q
	\right\}
	:
	\binom{R}{w}\!\ge M,\,
	Mw\!\le RK_{\max}
	\right\}.
\end{equation}
Then, the
proposed construction searches over $w\in{\mathcal W}_R$ and performs
several greedy constructions with random selection in \eqref{eq:greedy_load_metric}. Let $N_{\rm rst}$ denote the number of random
restarts.
If no complete
family $\cal F$ is found, the search should restart with a larger $R$. If a complete $\cal F$ is returned, the obtained codewords can inherently satisfy the directional observation requirement and concurrency limit.

In particular, the proposed scheme has a deterministic
worst-case fallback. When
$R=Mq$ and $w=q$,
a feasible schedule is obtained by assigning ${\cal F}=\{{\cal A}_2,\cdots,{\cal A}_N\}$, where
\begin{equation}
	{\cal A}_n
	=
	\{(n-2)q+1,\ldots,(n-1)q\},
	~~ \forall n\in {\cal N}/\{1\}.
\end{equation}
These supports are non-overlapping, and hence
$
	|{\cal A}_{n_1}\cap {\cal A}_{n_2}|=0=w-q,
	 \forall n_1\neq n_2,
$
which gives
$
d_H(\bm\chi_{n_1},\bm\chi_{n_2})=2q
$.
Moreover, every round contains exactly one transmitting node, so that
$
K_r=1\le K_{\max}
$
holds. Therefore, a feasible schedule always exists
for $R=Mq$. Together with the
lower bound, the achievable round count $R_{\rm alg}$ is bounded by
\begin{equation}
	R_{\rm cw,LB}
	\le
	R_{\rm alg}
	\le
	Mq,
\end{equation}
which restricts the search range of the round count.

\subsection{Computational Complexity}
\label{subsec:complexity}
We compare the computational complexity of the scheduling algorithm and its extension with the exhaustive search schemes. For a fixed pair $(R,w)$, let
$B\triangleq\binom{R}{w}$ denote the number of all candidate codewords.
After removing the permutations of node indices, the search space only
contains $\binom{B}{M}$ families. Checking
one family requires $\mathcal O(Mw+R)$ operations for $q=1$, whereas the Hamming distance calculations for $q>1$ require $\mathcal O(M^2w)$. The complexities of the exhaustive search schemes are derived as
\begin{equation}
	\begin{aligned}
		C_{\rm ex}^{(1)}
		&=\mathcal O\!\left(\binom{B}{M}(Mw+R)\right),\\
		C_{\rm ex}^{(q)}
		&=\mathcal O\!\left(\binom{B}{M}M^2w\right),~~q>1.
	\end{aligned}
	\label{eq:exhaustive_complexity}
\end{equation}

For Algorithm~\ref{alg:q1_exchange}, the parameter search over $R$ and $w_R$
requires at most $\mathcal O((R_{\rm cw}^{\star})^2)$ operations. The number of codeword exchanges is bounded by
$\mathcal O(M^2w)$. In each exchange, searching the selected supports and per-round loads costs $\mathcal O(M+R)$. Thus, the complexity of Algorithm~\ref{alg:q1_exchange} is
\begin{equation}
	C_{\rm prop}^{(1)} = \mathcal O\left(
	(R_{\rm cw}^{\star})^2+M^2w(M+R)
	\right).
	\label{eq:q1_algorithm_complexity}
\end{equation}
In this way, Algorithm~\ref{alg:q1_exchange} removes the combinatorial factor $\binom{B}{M}$ in \eqref{eq:exhaustive_complexity}, which dominates in the complexity of exhaustive search. 

For the extension in Algorithm~\ref{alg:qgreater1_adaptive}, 
each codeword-selection step scans at most $B$ candidates. The complexity of the full greedy construction for fixed
$(R,w)$ is therefore
$
	\mathcal O\left(M^2wB\right)
$.
Including the parameter search over $R$ and $w$, with $N_{\rm rst}$ random restarts, the computational complexity is
\begin{equation}
	\begin{split}
		C_{\rm prop}^{(q)} = \mathcal O\!\left(
		N_{\rm rst}M^2
		\sum_{\rho=R_{\rm cw,LB}}^{R_{\rm alg}}
		\sum_{w\in\mathcal W_{\rho}}
		w\binom{\rho}{w}
		+Mq\right).
	\end{split}
	\label{eq:qgreater1_algorithm_complexity}
\end{equation}
One can find that Algorithm~\ref{alg:qgreater1_adaptive} replaces the family enumeration,
which is approximately of order $B^{M}$, with only searching the candidate-codeword space. In this way, the constraints in ${\mathcal W}_R$, the early termination, and the $\mathcal O(Mq)$ fallback can reduce the practical computational load effectively.

\subsection{Discussion of Partially Connected Networks}
\label{subsec:partial_connectivity}
For a partially connected network where some nodes are not visible to each other, the proposed transmit--listen schedules, which are designed for full connectivity, are still applicable. Firstly, the links related to node 1 are required to be bidirectionally available for every node $n\in{\cal N}/\{1\}$. Otherwise, the node cannot achieve clock-scale compensation and thus will not be scheduled in the network ranging. Secondly, in this scenario, the ranges are only updated for available links, while the other ranges associated with missing links do not need to be estimated.

Using partial connections to further reduce rounds requires adjustments to the constraints of node links. Specifically, for a partially-connected network, let $\mathcal E$ denote the set of links requiring range updates between non-reference nodes. For a given $\mathcal E$, the node-pair requirements in \eqref{prob:dir_constraint} can be replaced by
\begin{equation}
\Gamma_{n_1\rightarrow n_2}\geq q,\quad
\Gamma_{n_2\rightarrow n_1}\geq q,\quad
\forall\{n_1,n_2\}\in\mathcal E.
\end{equation}
To enable range estimation with other nodes, each participating node should transmit at least once. Under these requirements and the same concurrency limit, the resulting minimum $R$ will not exceed that for full connectivity. Consequently, the proposed scheduling algorithm can still work in partially-connected networks.

\section{Numerical Simulation Results}\label{sec:nsr}

In this section, numerical simulations are carried out to verify the effectiveness of the proposed schemes. We specify the simulation setup in Section~\ref{subsec:simulation_setup} and compare schedule construction efficiency in Section~\ref{subsec:construction_time}. The round counts and ranging accuracy for $q=1$ and $q>1$ are evaluated in Sections~\ref{subsec:simulation_q1} and~\ref{subsec:simulation_qgt1}, respectively. We examine the effects of the concurrency limit, redundancy, network size, and time budget in Section~\ref{subsec:simulation_tradeoffs}. Finally, the relative-localization accuracy is evaluated in Section~\ref{subsec:relative_localization}.

\begin{table*}[!t]
\centering
\caption{\normalsize Construction times of schedules in the constant-weight codeword class.}
\label{tab:construction_representative}
\begingroup
\normalsize
\setlength{\tabcolsep}{4pt}
\renewcommand{\arraystretch}{1.12}
\begin{tabular*}{\textwidth}{@{\extracolsep{\fill}} c c c c r r r @{}}
\toprule
Algorithm & $(N,q,K_{\max})$ & $(R_{\rm ex}^{\star},w^*)$ & $(\widehat R,\widehat w)$ & Construction (ms) & Exhaustive (ms) & Speedup ($\times$) \\
\midrule
1 & $(40,1,10)$ & $(10,2)$ & $(10,2)$ & 0.273 & 0.399 & 1.46 \\
1 & $(60,1,15)$ & $(12,2)$ & $(12,2)$ & 0.681 & 1.007 & 1.48 \\
1 & $(100,1,25)$ & $(12,3)$ & $(12,3)$ & 1.342 & 13754.607 & 10253.01 \\
1 & $(200,1,50)$ & Not obtained & $(12,3)$ & 1.742 & Timeout (30 s) & N/A \\
\midrule
3 & $(12,2,3)$ & $(11,3)$ & $(11,3)$ & 5.036 & 57.910 & 11.50 \\
3 & $(14,2,3)$ & $(13,3)$ & $(13,3)$ & 5.910 & 235.815 & 39.90 \\
3 & $(16,2,4)$ & $(12,3)$ & $(12,3)$ & 5.281 & 27.922 & 5.29 \\
3 & $(20,2,5)$ & $(12,3)$ & $(12,3)$ & 6.217 & 17.417 & 2.80 \\
3 & $(40,2,10)$ & Not obtained & $(16,4)$ & 44.507 & Timeout (30 s) & N/A \\
3 & $(100,2,25)$ & Not obtained & $(16,4)$ & 53.945 & Timeout (30 s) & N/A \\
\midrule
3 & $(4,3,2)$ & $(9,3)$ & $(9,3)$ & 6.268 & 7.238 & 1.15 \\
3 & $(8,3,2)$ & $(14,4)$ & $(15,4)$ & 15.341 & 4294.739 & 279.96 \\
3 & $(10,3,3)$ & $(12,4)$ & $(15,4)$ & 14.900 & 128.002 & 8.59 \\
3 & $(14,3,4)$ & $(13,4)$ & $(16,4)$ & 24.095 & 382.457 & 15.87 \\
3 & $(15,3,4)$ & $(14,4)$ & $(16,4)$ & 18.463 & 1387.612 & 75.16 \\
3 & $(16,3,4)$ & $(15,4)$ & $(16,4)$ & 14.141 & 8782.969 & 621.08 \\
3 & $(17,3,4)$ & $(16,4)$ & $(16,4)$ & 8.182 & 48906.002 & 5976.98 \\
3 & $(40,3,10)$ & Not obtained & $(20,5)$ & 142.972 & Timeout (60 s) & N/A \\
3 & $(100,3,25)$ & Not obtained & $(23,5)$ & 2016.052 & Timeout (60 s) & N/A \\
\bottomrule
\end{tabular*}
\par\vspace{4pt}
\begin{minipage}{\textwidth}
\footnotesize
$(R_{\rm ex}^{\star},w^*)$ and $(\widehat R,\widehat w)$ denote the schedules returned by exhaustive search and the proposed algorithm, respectively. Speedup is the ratio of the median exhaustive-search time to the median construction time. ``Timeout (30 s)'' and ``Timeout (60 s)'' indicate that exhaustive search reached the corresponding time limit without returning a result. In these cases, the exhaustive-search schedule was not obtained and the speedup is unavailable (N/A).
\end{minipage}
\endgroup
\end{table*}

\subsection{Simulation Setup}

\label{subsec:simulation_setup}
Unless otherwise specified, we consider $N=100$ nodes, which are uniformly distributed within a $200\times200$~m$^2$ region. Clock-frequency deviations are uniformly distributed within
$[-40,40]$~ppm. Each round occupies one reserved slot of duration $\Delta=1$~ms, including transmission, processing, and guard time. In the simulations, the distribution of ${\rv w}_\alpha^{(n)}$ in \eqref{eq:clock_model} depends on whether node $n$ transmits or receives in event $\alpha$. We denote the corresponding transmission (TX) and reception (RX) timestamp errors by ${\rv w}_{\rm TX}$ and ${\rv w}_{\rm RX}$, respectively, and set
\begin{equation}
	{\rv w}_{\rm TX}\sim\mathcal N(0,\sigma^2)~~\text{and}~~
	{\rv w}_{\rm RX}\sim\mathcal N(0,K_r\sigma^2),
	\label{eq:sim_noise}
\end{equation}
where $K_r$ is the actual number of simultaneous transmitters in round $r$, and $\sigma$ is the single-transmitter timestamp-error standard deviation. The RX timestamp-error variance scales linearly with $K_r$, while the TX timestamp-error variance remains $\sigma^2$. This model captures the $\sqrt{K_r}$-scaling law of timestamp-error standard-deviation in time-division or code-division ranging \cite{SheWin:J10a}. The default noise level is $\sigma=1$~ns. The following two ranging schemes are compared in simulations.

\begin{enumerate}
	\item {\bf SM-NR \cite{zhang2022smnr} (baseline)}: The existing SM-NR scheme in \cite{zhang2022smnr}, including its protocol and ToF estimator, is adopted for network ranging as a baseline. It requires $N+1$ rounds in each ranging cycle.
	\item {\bf CC-SM-NR (proposed)}: The proposed CC-SM-NR scheme is adopted for ranging, which costs $R+2$ rounds in each cycle. Algorithm~\ref{alg:q1_exchange} is used for schedule design when $q=1$. For $q>1$, the extension in Algorithm~\ref{alg:qgreater1_adaptive} is used with $N_{\rm rst}=20$ random restarts. The ML estimator in \eqref{eq:ToF_estimator} is used for range measurements.
\end{enumerate}
The metric to evaluate the ranging accuracy is the root mean square error (RMSE), defined as
\begin{equation}
	\operatorname{RMSE}=
	\sqrt{{\mathbb E}\left\{ \frac{2}{N(N\!-\!1)}
		\sum_{n_1<n_2}
		(\widehat d_{n_1,n_2}-d_{n_1,n_2})^2
		\right\}
		},
	\label{eq:sim_metric}
\end{equation}
where $\widehat d_{n_1,n_2}$ is the estimation of $d_{n_1,n_2}$ from \eqref{eq:weighted_ToF}. For all schemes, we consider the averaged performance of $10^4$ independent topology and clock state trials in simulations.

\subsection{Schedule Construction Time}
\label{subsec:construction_time}

In Table~\ref{tab:construction_representative}, we compare the construction times with exhaustive search over the same constant-weight codeword class. The simulations use MATLAB R2024b on an AMD Ryzen 9 7945HX CPU with 63.7~GiB RAM. Speedup is the ratio of the median exhaustive-search time to the median construction time. Compared with exhaustive search, the scheduling algorithm and its extension can reduce the construction time by avoiding enumeration of complete codeword families. In these cases, speedups range from $1.5\times$ to $10253.0\times$ for Algorithm~\ref{alg:q1_exchange} and from $1.2\times$ to $5977.0\times$ for Algorithm~\ref{alg:qgreater1_adaptive}. For $q=1$, Algorithm~\ref{alg:q1_exchange} attains the minimum round count within the constant-weight class. For $q>1$, the greedy schedule design can require additional rounds, with a largest gap of three rounds among the cases certified by exhaustive search, e.g., $(N,q,K_{\max})=(10,3,3)$. These results demonstrate the complexity analysis in Section~\ref{subsec:complexity}.

\subsection{Scheduling and Accuracy when $q=1$}
\label{subsec:simulation_q1}
\begin{figure}[!tbp]
	\setlength{\abovecaptionskip}{0pt}
	\centering\includegraphics[width=\columnwidth,trim=0 9.75bp 0 0,clip]{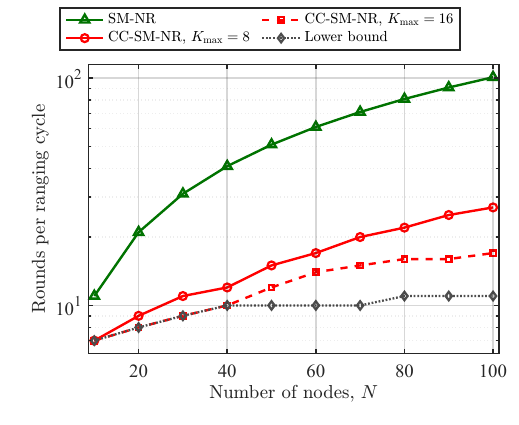}
	\caption{Occupied rounds per ranging cycle for different schemes when $q=1$.}
	\label{fig:sim_scaling_q1}
\end{figure}

\begin{figure}[!tbp]
	\setlength{\abovecaptionskip}{0pt}
	\centering\includegraphics[width=\columnwidth,trim=0 9.25bp 0 0,clip]{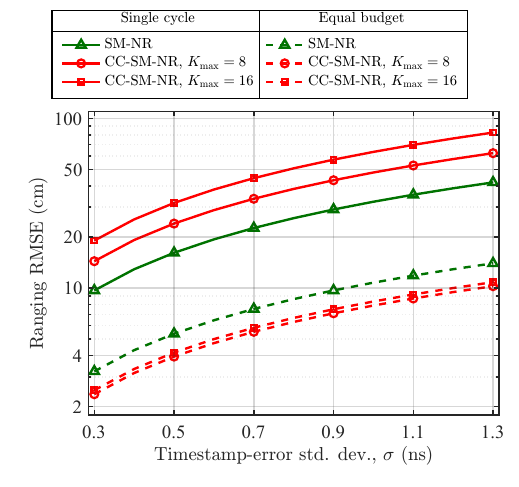}
	\caption{Ranging RMSE as a function of timestamp-error standard deviation $\sigma$ for $q=1$.}
	\label{fig:sim_q1_budget}
\end{figure}

We consider $K_{\max}\in\{8,16\}$. Let $L_s$ denote the number of occupied rounds in one ranging cycle, where $s\in\{\mathrm{SM},\mathrm{CC}\}$ identifies SM-NR and CC-SM-NR, respectively. Including the two reference broadcasts, the round counts are
\begin{equation}
 L_{\rm SM}=Mq+2,\qquad L_{\rm CC}=R+2,
 \label{eq:sim_costs}
\end{equation}
where $M=N-1$. Thus, $L_s$ equals $L_{\rm SM}$ or $L_{\rm CC}$ for the corresponding protocol, and each cycle lasts $L_s\Delta$. For $q=1$, $L_{\rm SM}=N+1$. The extension for $q>1$ is specified in \eqref{eq:sim_smnr_q}.

Fig.~\ref{fig:sim_scaling_q1} shows that the concurrent scheduling reduces the occupied rounds as the network size grows. The lower bound is $R_\infty^\star+2$, obtained from Theorem~\ref{thm:sperner}. For $K_{\max}=16$, the constructed schedules attain this bound over $10\le N\le40$.
For larger $N$, more transmissions are needed. Due to the per-round concurrency limit, additional rounds are required even for an optimal constant-weight schedule. With fixed $K_{\max}$, the logarithmic growth does not hold as $N$ increases.

Fig.~\ref{fig:sim_q1_budget} compares single-cycle and equal-budget accuracy. The former uses one complete cycle per scheme. The latter allows $T_{\rm budget}=1000$~ms and averages the available ToF estimates over $J_s=\lfloor T_{\rm budget}/(L_s\Delta)\rfloor$ complete cycles for each node pair. One can observe that
the RMSE increases approximately linearly with $\sigma$. At $\sigma=1$~ns, the single-cycle ranging RMSEs are $32.3$~cm, $48.0$~cm, and $63.5$~cm for SM-NR and CC-SM-NR with $K_{\max}=8$ and $16$, respectively. This increase reflects the larger RX timestamp-error variance associated with concurrent transmissions under \eqref{eq:sim_noise}. The corresponding cycle durations are $101$~ms, $27$~ms, and $17$~ms, allowing $9$, $37$, and $58$ complete cycles within 1000~ms. The shorter cycles provide more measurements for averaging, reducing the equal-budget ranging RMSEs to $10.8$~cm, $7.9$~cm, and $8.3$~cm. Despite its larger single-cycle errors, CC-SM-NR reduces the equal-budget ranging RMSE by approximately $27\%$ and $23\%$. A shorter cycle improves equal-budget accuracy only if the extra averaging compensates for its larger single-cycle error. This explains why $K_{\max}=16$ gives a shorter cycle but a higher RMSE than $K_{\max}=8$.

\subsection{Scheduling and Accuracy when $q>1$}
\label{subsec:simulation_qgt1}

We consider $q\in\{2,3\}$ and $K_{\max}\in\{8,16\}$. To provide $q$ observations in each direction between non-reference nodes, the SM-NR baseline should use the transmission order
\begin{equation}
 1,\underbrace{2,\ldots,2}_{q},
 \underbrace{3,\ldots,3}_{q},\ldots,
 \underbrace{N,\ldots,N}_{q},1,
 \label{eq:sim_smnr_q}
\end{equation}
where each number identifies a transmitting node. The two reference broadcasts enclose all non-reference transmissions, which leads to $Mq+2$ rounds per cycle.

In Fig.~\ref{fig:sim_scaling_qgt1}, we show that additional observations lengthen the ranging cycle, while concurrent scheduling retains a substantial reduction in occupied rounds. At $N=100$, the serial cycles last $200$~ms and $299$~ms for $q=2$ and $3$, respectively. For $K_{\max}=8$ and $16$, the constructed CC-SM-NR cycles last $40$~ms and $27$~ms when $q=2$, and $52$~ms and $33$~ms when $q=3$. Concurrent transmissions therefore accommodate the increased observation requirement within shorter cycles.

\begin{figure}[!tbp]
	\setlength{\abovecaptionskip}{0pt}
	\centering\includegraphics[width=\columnwidth,trim=0 9.75bp 0 0,clip]{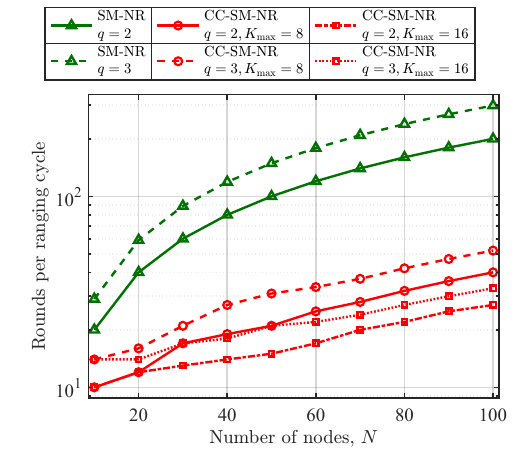}
	\caption{Rounds per ranging cycle for different schemes when $q\in\{2,3\}$.}
	\label{fig:sim_scaling_qgt1}
\end{figure}

\begin{figure}[!tbp]
	\setlength{\abovecaptionskip}{0pt}
	\centering\includegraphics[width=\columnwidth,trim=0 9.25bp 0 0,clip]{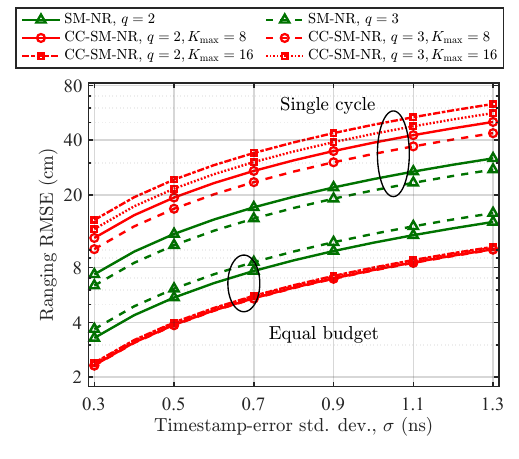}
	\caption{Ranging RMSE as a function of timestamp-error standard deviation $\sigma$ for $q\in\{2,3\}$.}
	\label{fig:sim_qgt1}
\end{figure}

Fig.~\ref{fig:sim_qgt1} shows that higher redundancy improves single-cycle accuracy by providing more observations to combine. At $\sigma=1$~ns, the single-cycle RMSE of SM-NR decreases from $24.4$~cm to $21.2$~cm as $q$ increases from 2 to 3. Under the common 1000-ms budget, however, these longer serial cycles permit only 5 and 3 complete measurements. The CC-SM-NR completes $25$ and $37$ cycles for $q=2$, and $19$ and $30$ cycles for $q=3$, at $K_{\max}=8$ and $16$, respectively. For $q=2$, the equal-budget ranging RMSEs are $7.7$~cm and $8.0$~cm for CC-SM-NR, compared with $10.9$~cm for SM-NR. For $q=3$, they are $7.7$~cm and $7.9$~cm, compared with $12.2$~cm. These correspond to reductions of approximately $29\%$ and $27\%$ for $q=2$, and $37\%$ and $35\%$ for $q=3$. Observations within one cycle share the same reference timestamps, whereas each new cycle provides new reference measurements. Therefore, increasing $q$ and repeating more cycles do not necessarily give the same accuracy within a fixed time budget.

\subsection{Concurrency Limit, Redundancy, and Network Size}
\label{subsec:simulation_tradeoffs}

\begin{figure}[!tbp]
	\setlength{\abovecaptionskip}{0pt}
	\centering\includegraphics[width=\columnwidth,trim=0 6.25bp 0 0,clip]{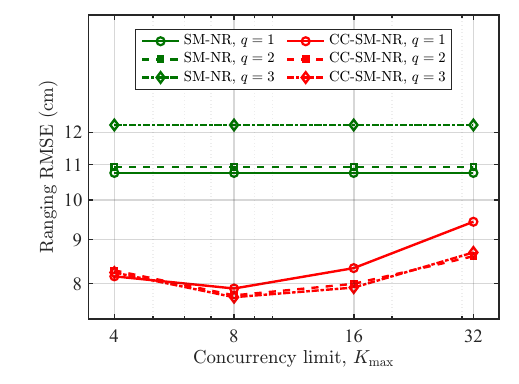}
	\caption{Effect of redundancy and the concurrency limit under an equal time budget $T_{\rm budget}$.}
	\label{fig:sim_qk}
\end{figure}

\begin{figure}[t]
	\setlength{\abovecaptionskip}{0pt}
	\centering\includegraphics[width=\columnwidth,trim=0 9.75bp 0 0,clip]{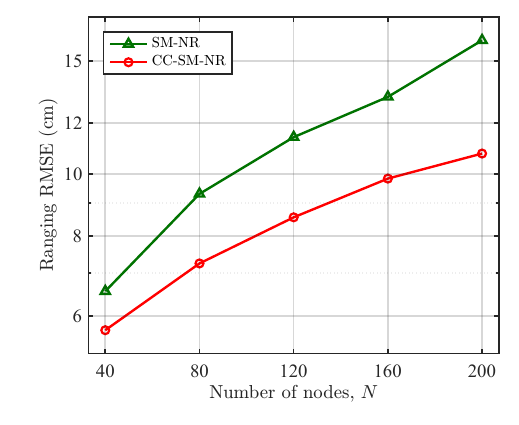}
	\caption{Equal-budget ranging RMSE as a function of $N$ for $q=1$.}
	\label{fig:sim_nodes}
\end{figure}

\begin{figure}[!tbp]
	\setlength{\abovecaptionskip}{0pt}
	\centering\includegraphics[width=\columnwidth,trim=0 5.75bp 0 0,clip]{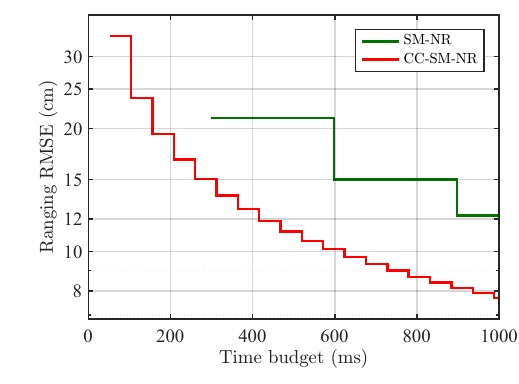}
	\caption{Ranging RMSE versus time budget $T_{\rm budget}$ for $q=3$.}
	\label{fig:sim_budget}
\end{figure}

Fig.~\ref{fig:sim_qk} compares the ranging RMSE for different $q$ and $K_{\max}$ under a common 1000-ms budget. One can observe that the CC-SM-NR achieves a lower RMSE than the baseline in every case, and its RMSE first decreases as $K_{\max}$ increases from 4 to 8 and then increases at larger concurrency limits. In particular, $K_{\max}=8$ gives the lowest RMSEs of $7.9$~cm, $7.7$~cm, and $7.7$~cm for $q=1,2,3$, respectively. It implies that, although a higher concurrency limit permits more complete cycles within the budget, the RX timestamp-error variance increases with the number of simultaneous transmitters under \eqref{eq:sim_noise}. This result reflects the balance between these two effects. Besides, at $K_{\max}=8$, increasing $q$ improves the RMSE only slightly. This is because the additional observations within a cycle are partly offset by fewer completed cycles. 

Fix the concurrency limit at $K_{\max}=8$. In Fig.~\ref{fig:sim_nodes}, we show the effect of network size for $q=1$. The ranging RMSE increases with $N$ for both schemes, while CC-SM-NR maintains a lower error over the entire tested range. Although a larger network has more node pairs, its longer ranging cycle leaves fewer measurements of each pair for averaging within the same time budget. The CC-SM-NR scheme can reduce this loss by shortening the cycle. At $N=200$, the RMSE is $10.8$~cm for CC-SM-NR, compared with $16.2$~cm for SM-NR.

In Fig.~\ref{fig:sim_budget}, we further compare the ranging RMSE as the time budget increases. The flat intervals show that extending the time budget helps only when another complete cycle is obtained. Shorter cycles therefore provide more frequent opportunities to improve the estimate. The CC-SM-NR scheme provides its first complete network measurement after 52~ms, whereas the serial SM-NR baseline requires 299~ms. Although the first CC-SM-NR cycle has a larger ranging error, its shorter duration permits more frequent averaging updates. Within 1~s, CC-SM-NR completes 19 cycles and achieves an RMSE of $7.7$~cm, whereas SM-NR completes three cycles and achieves $12.2$~cm. The $37\%$ reduction demonstrates the accuracy benefit of acquiring more measurements within the same time budget.

\subsection{Relative Localization}
\label{subsec:relative_localization}

\begin{figure}[!t]
	\setlength{\abovecaptionskip}{0pt}
	\centering\includegraphics[width=\columnwidth,trim=0 5.75bp 0 0,clip]{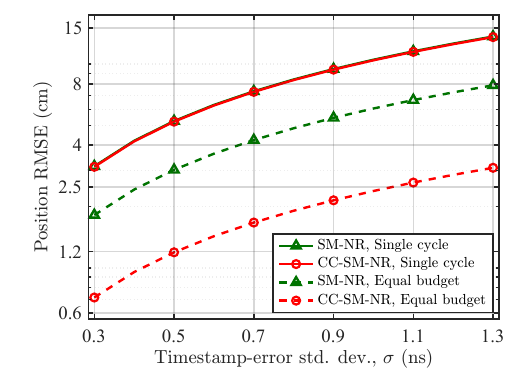}
	\caption{Position RMSE as a function of timestamp-error standard deviation $\sigma$ for $q=3$.}
	\label{fig:sim_loc}
\end{figure}
To evaluate the localization accuracy of different schemes, following \cite{zhang2022smnr}, we adopt a multidimensional scaling (MDS) initialization, as well as a weighted scaling by majorizing a complicated function (SMACOF) method, to estimate the relative locations of multiple nodes.  In Fig.~\ref{fig:sim_loc}, the relative-localization accuracy of different schemes is compared by fixing $q=3$ and $K_{\max}=8$. As shown in Fig.~\ref{fig:sim_loc}, the position RMSE increases with the timestamp-error standard deviation $\sigma$ for both schemes, and their single-cycle curves nearly coincide. Under the common 1000-ms budget, however, the corresponding RMSEs decrease to $6.1$ and $2.4$~cm at $\sigma=1$~ns. The larger number of complete ranging cycles allows CC-SM-NR to supply more accurate distance estimates, obtained by averaging across cycles, to the localization algorithm. The similar single-cycle position RMSEs, despite different ranging RMSEs, show that ranging RMSE alone does not determine localization accuracy. The localization algorithm uses all pairwise distances, so network geometry and correlations between ranging errors also affect the position estimates.

\begin{figure}[!t]
	\setlength{\abovecaptionskip}{0pt}
	\centering\includegraphics[width=0.9\columnwidth]{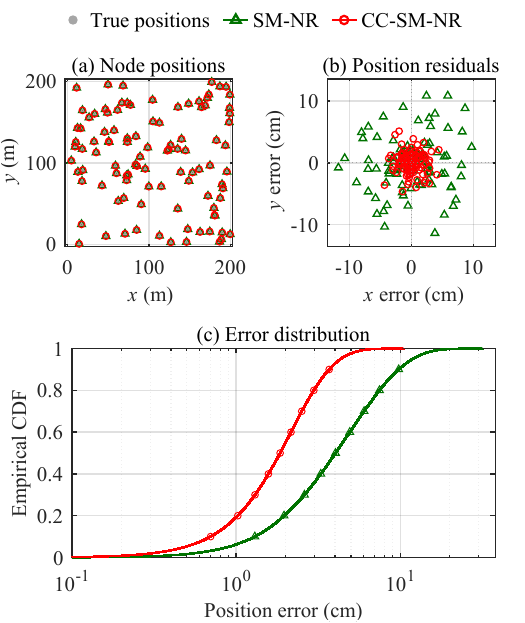}
	\caption{Equal-budget localization for $q=3$: (a) the true and estimated node positions, (b) the position errors of all nodes, and (c) the cumulative distribution function (CDF) of position errors.}
	\label{fig:sim_topology}
\end{figure}

In Fig.~\ref{fig:sim_topology}, the equal-budget localization results are shown. In Fig.~\ref{fig:sim_topology}(a), we find that both sets of estimated positions are close to the true positions at the scale of the deployment region. The position errors in Fig.~\ref{fig:sim_topology}(b) reveal the difference more clearly. The proposed CC-SM-NR scheme produces a denser distribution around the origin than the SM-NR scheme. In Fig.~\ref{fig:sim_topology}(c), the CDF of position errors is also shifted toward smaller errors for CC-SM-NR. In particular, the 90th-percentile position error decreases from $9.8$~cm for SM-NR to $3.7$~cm for CC-SM-NR, complementing the RMSE reduction shown in Fig.~\ref{fig:sim_loc}. The CDF shows that, in most cases, the CC-SM-NR can reduce the position errors compared to SM-NR.


\section{Conclusions}\label{sec:con}
This paper established the CC-SM-NR framework for asynchronous half-duplex networks. The CC-SM-NR protocol combines the transmit--listen rounds with the clock-scale compensation, enabling reciprocal range measurements under a finite concurrency limit. Our analysis determines the exact minimum number of transmit--listen rounds and its minimum within the constant-weight codeword class under a concurrency limit. The scheduling algorithm combines parameter search with load-balancing exchanges to achieve the minimum number of transmit--listen rounds. The greedy extension of this algorithm provides feasible schedules for higher observation redundancy. The simulations demonstrate that, for a given time budget, the shorter ranging cycle of CC-SM-NR enables it to achieve higher accuracy compared with the SM-NR scheme.

For future work, one may consider heterogeneous networks where selected full-duplex nodes could reduce ranging rounds, subject to residual self-interference. Combining network ranging with integrated sensing and communication (ISAC) could support environmental sensing through shared waveforms. Multi-antenna nodes could provide angle measurements to complement ranges and improve estimates of network geometry. Novel antenna arrays, such as reconfigurable intelligent surfaces (RISs) and movable antennas, also motivate joint optimization of surface configurations, antenna positions, and ranging schedules. These extensions should be assessed under common time budgets, including the overhead of hardware control and antenna reconfiguration.

\appendices
\section{Proof of Lemma~\ref{lem:almost_regular_cw}}
\label{app:proof_almost_regular_cw}

\begin{IEEEproof}
	For a given collection $\mathcal G$, let
	\begin{equation}
		K_r(\mathcal G)
		\triangleq
		\left|\left\{{\cal A}\in\mathcal G:r\in {\cal A}\right\}\right|
	\end{equation}
	be the number of selected codewords containing coordinate $r$. Since
	$M\leq\binom{R}{w}$, at least one such collection exists. One can find one $\mathcal G$ that minimizes the imbalance potential function
	\begin{equation}
		\Psi(\mathcal G)
		\triangleq
		\sum_{r=1}^{R}K_r^2(\mathcal G).
		\label{eq:load_potential}
	\end{equation}
	For contradiction, suppose that there are two rounds $a$ and $b$, which satisfy
	\begin{equation}
		K_a\geq K_b+2.
		\label{eq:unbalanced_pair}
	\end{equation}
	Then, we define the two codeword subfamilies
	\begin{align}
		\mathcal G_{a\bar b}
		&\triangleq
		\left\{{\cal A}\in\mathcal G:a\in {\cal A},\ b\notin {\cal A}\right\},\\
		\mathcal G_{\bar a b}
		&\triangleq
		\left\{{\cal A}\in\mathcal G:a\notin {\cal A},\ b\in {\cal A}\right\}.
	\end{align}
	The codewords containing both round $a$ and round $b$ contribute equally to $K_a$
	and $K_b$, so we have
	\begin{equation}
		|\mathcal G_{a\bar b}|-|\mathcal G_{\bar a b}|
		=
		K_a-K_b
		>0.
		\label{eq:exchange_cardinality}
	\end{equation}
	For every ${\cal A}\in\mathcal G_{a\bar b}$, define the exchanged weight-$w$ set as
	\begin{equation}
		{\cal A}'
		\triangleq
		({\cal A}/\{a\})\cup\{b\}.
		\label{eq:exchange_set}
	\end{equation}
	If every such ${\cal A}'$ already belonged to $\mathcal G$, the mapping from
	${\cal A}$ to ${\cal A}'$ would be an injection from $\mathcal G_{a\bar b}$ into
	$\mathcal G_{\bar a b}$. This would imply
	$|\mathcal G_{a\bar b}|\leq|\mathcal G_{\bar a b}|$, contradicting the supposition
	\eqref{eq:exchange_cardinality}. Hence, there exists at least one
	${\cal A}\in\mathcal G_{a\bar b}$ for which ${\cal A}'\notin\mathcal G$. By replacing ${\cal A}$ by ${\cal A}'$, we obtain
	\begin{equation}
		\mathcal G'
		=
		(\mathcal G/\{{\cal A}\})\cup\{{\cal A}'\}.
	\end{equation}
	The new collection $\mathcal G'$ still contains $M$ distinct weight-$w$ codeword sets. Moreover,
	only the loads of rounds $a$ and $b$ are updated as follows:
	\begin{equation}
	K_a'=K_a-1 \quad \text{and} \quad  K_b'=K_b+1.
	\end{equation}
	The corresponding change in the potential function is
	\begin{align}
	\Psi(\mathcal G')-\Psi(\mathcal G)
	=&
	(K_a-1)^2+(K_b+1)^2-K_a^2-K_b^2 \notag \\
	=&
	-2(K_a-K_b-1)
	<0,
	\end{align}
	which contradicts the assumed minimality of $\Psi(\mathcal G)$. Therefore,
	no round pair satisfying \eqref{eq:unbalanced_pair} can exist, which proves
	\eqref{eq:almost_regular_load}.
	Finally, because every selected codeword has weight $w$,
	$
		\sum_{r=1}^{R}K_r=Mw
	$ holds. Combining this equality with \eqref{eq:almost_regular_load}, we obtain the
	 distribution in \eqref{eq:load_division} and the peak in
	\eqref{eq:optimal_peak_load}.
\end{IEEEproof}

\footnotesize
\bibliographystyle{IEEEtran}
\bibliography{IEEEabrv,SGroupDefinition,SGroup}
\balance

\begin{IEEEbiography}[{\includegraphics[width=1in,height=1.25in,trim=72.74091bp 102.73716bp 75.85302bp 5.99925bp,clip,keepaspectratio]{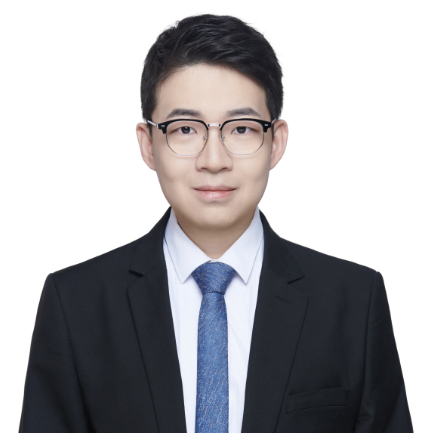}}]{Zijian Zhang}
	(Graduate Student Member, IEEE) received the B.E. degree and Ph. D. degree in electronic engineering from Tsinghua University, Beijing, China, in 2020 and 2025, respectively. His research area is network localization and beyond massive MIMO for future wireless systems. He has received Tsinghua Science Technology Best Paper Award in 2024 and the IEEE Signal Processing Society Young Author Best Paper Award in 2025. In 2024, he won the Special Scholarship of Tsinghua University, which annually awards 10 students out of over 40,000 graduate students in Tsinghua University. He was listed in Standford World's Top 2\% Scientists in 2023 and 2025.
\end{IEEEbiography}

\begin{IEEEbiography}
    [{\includegraphics[width=1in,height=1.25in,clip,keepaspectratio]{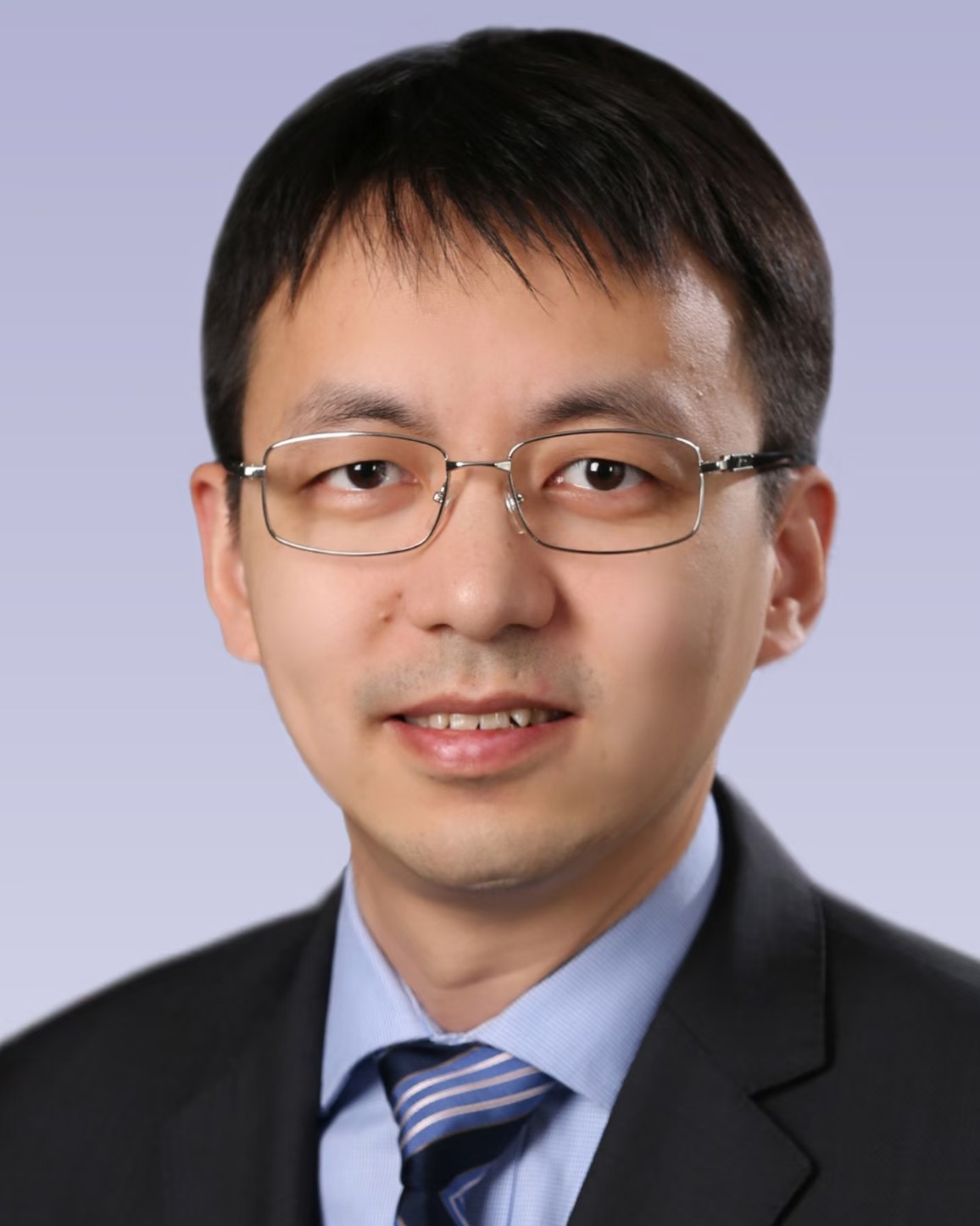}}]
{Yuan Shen} (Senior Member, IEEE) received the B.E. degree in electronic engineering from Tsinghua University in 2005, and the S.M. and Ph.D. degrees in electrical engineering and computer science from the Massachusetts Institute of Technology (MIT) in 2008 and 2014, respectively. He is currently a Full Professor with the Department of Electronic Engineering, Tsinghua University. His research interests include integrated sensing and communication, multi-agent systems and AI for science. His papers have received the IEEE ComSoc Fred W. Ellersick Prize and several best paper awards from IEEE conferences. He has served as the TPC Symposium Co-Chair for IEEE ICC and IEEE Globecom for several times. He was the Elected Chair for the IEEE ComSoc Radio Communications Committee from 2019 to 2020. He is currently an Editor of the IEEE \textsc{TRANSACTIONS ON COMMUNICATIONS}, IEEE \textsc{TRANSACTIONS ON WIRELESS COMMUNICATIONS}, and \textsc{China Communications}.
\end{IEEEbiography}

\end{document}